\documentclass[smallcondensed]{svjour3}

\usepackage[utf8]{inputenc}
\usepackage[english]{babel}
\usepackage[T1]{fontenc}
\usepackage{amsmath}
\usepackage{amsfonts}
\usepackage{amssymb}
\usepackage{makeidx}
\usepackage{array}
\usepackage{color}
\usepackage{relsize}
\usepackage{verbatim}
\usepackage{graphicx} 
\usepackage[table]{xcolor}
\usepackage[all]{xy}
\usepackage{enumitem}
\usepackage{algorithm}
\usepackage{algorithmic}
\usepackage{ifthen}

\smartqed

\newtheorem{prop}		{Proposition}
\newtheorem{lm}[prop]{Lemma}

\renewcommand{\leq}{\leqslant}
\renewcommand{\geq}{\geqslant}
\renewcommand{\le}{\leqslant}
\renewcommand{\ge}{\geqslant}

\newcommand{\Oubliettes}[1]{}

\newcommand{\enlevediapo}[1]{}

\definecolor{gris-ligne-vide}{gray}{0.90}

\newcommand{\ION}{
\renewcommand{\algorithmicrequire}{\textbf{Input:}}
\renewcommand{\algorithmicensure}{\textbf{Output:}}
}
\newcommand{\IOB}{
\renewcommand{\algorithmicrequire}{\textcolor{white}{\textbf{Input:}}}
\renewcommand{\algorithmicensure}{\textcolor{white}{\textbf{Output:}}}
}

\newcommand{\set}[1]{\left\{#1\right\}}
\newcommand{\size}[1]{\left|#1\right|}

\DeclareMathOperator{\F}		{\mathbb{F}}
\DeclareMathOperator{\Q}		{\mathbb{Q}}
\DeclareMathOperator{\Z}		{\mathbb{Z}}
\DeclareMathOperator{\N}		{\mathbb{N}}
\newcommand{\OK}	{\mathcal{O}_K}
\newcommand{\OL}	{\mathcal{O}_L}

\newcommand{\OF}	{\mathcal{O}_F}
\DeclareMathOperator{\OGab}	{\mathcal{O}_{\Gab}}
\DeclareMathOperator{\OC}	{\mathcal{O}_{\mathcal C}}
\DeclareMathOperator{\M}		{\mathcal{M}}

\DeclareMathOperator{\roots}	{Roots}

\DeclareMathOperator{\Auto}	{Aut}
\DeclareMathOperator{\Vect}	{Vect}
\DeclareMathOperator{\Ann}	{\mathcal{A}}
\DeclareMathOperator{\Int}	{\mathcal{I}}

\DeclareMathOperator{\rang}	{rank}
\DeclareMathOperator{\Deg}	{deg}
\DeclareMathOperator{\Dim}	{dim}
\DeclareMathOperator{\Min}	{min}
\DeclareMathOperator{\Gab}	{Gab}
\DeclareMathOperator{\Dec}	{Dec}
\DeclareMathOperator{\NLR}	{NLR}
\DeclareMathOperator{\LR}	{LR}
\DeclareMathOperator{\Hdim}	{\mathcal{H}_{dim}}

\newcommand{\wb}	{w_{\mathcal{B}}}
\newcommand{\wa}	{w_{\mathcal{A}}}
\newcommand{\wthK}	{w_{\theta,K}}
\newcommand{\wthL}	{w_{\theta,L}}

\newcommand{\db}	{d_{\mathcal{B}}}
\newcommand{\da}	{d_{\mathcal{A}}}
\newcommand{\dthK}	{d_{\theta,K}}
\newcommand{\dthL}	{d_{\theta,L}}

\newcommand{\ceil}[1]{\left\lceil#1\right\ceil}

\renewcommand{\deg}			{\Deg}
\renewcommand{\dim}			{\Dim}
\renewcommand{\min}			{\Min}

\newcommand{\TA}{Welch-Berlekamp like algorithm}
\newcommand{\vectgr}[1]{\vec{#1}}
\newcommand{\vecteurgras}[1]{\vec{#1}}
\newcommand{\matrice}[1]{\vec{#1}}

\newcommand{\Monome}[2]{
	\ifthenelse%
		{\equal{#2}{00}}%
		{{}}%
		{\ifthenelse%
			{\equal{#2}{0}}%
			{1}%
			{\ifthenelse%
				{\equal{#2}{1}}%
				{#1}%
				{{#1}^{#2}}%
			}%
		}%
}%

\newcommand{\AppliqueTheta}[2]	{\theta^{#2}({#1})}
\newcommand{\Eval}[2]			{{#1}\left\{ #2 \right\}}

\newcommand{\Annul}[1]			{\Ann_{#1}}
\newcommand{\Interpol}[2]		{\Int_{#1,#2}}

\newcommand{\ForallInt}[3]		{,#1 \leq #2 \leq #3}

\newcommand{\VecteurLigneCoins}[2]	{ (#1, \ldots, #2) }
\newcommand{\VecteurCoins}[2]{
		\left( \begin{array}{c}
		#1 \\ 
		\vdots \\ 
		#2 \\ 
		\end{array} \right) 
		}
\newcommand{\MatriceCoins}[4]{
		\left( \begin{array}{ccc}
		#1 & \cdots & #2 \\ 
		\vdots & \ddots & \vdots \\ 
		#3 & \cdots & #4 \\ 
		\end{array} \right) 
		}

\newcommand{\MatriceGeneraleEch}[3]{
		\left( \begin{array}{ccc}
		#1_{1,1} & \cdots & #1_{#3,1} \\ 
		\vdots & \ddots & \vdots \\ 
		#1_{1,#2} & \cdots & #1_{#3,#2} \\ 
		\end{array} \right) 
		}

\begin{document}

\title{Generalized Gabidulin codes over  fields of any characteristic}
\author{Daniel Augot \and Pierre Loidreau \and Gwezheneg Robert}

\institute{
	Daniel Augot 
		\at INRIA Saclay-\^Ile-de-France, and École polytechnique, Palaiseau, France 
		\\ \email daniel.augot@inria.fr
\and 
	Pierre Loidreau 
		\at DGA MI and IRMAR, Université de Rennes 1 
		\\ \email pierre.loidreau@univ-rennes1.fr
\and 
	Gwezheneg Robert  
		\at DGA MI
		\\ \email gwezheneg.robert@intradef.gouv.fr
}

\maketitle

\begin{abstract}
  We generalist Gabidulin codes to the case of infinite fields,
  eventually with characteristic zero.  For this purpose, we consider
  an abstract field extension and any automorphism in the Galois
  group.  We derive some conditions on the automorphism to be able to
  have a proper notion of rank metric which is in coherence with
  linearized polynomials.  Under these conditions, we generalize
  Gabidulin codes and provide a decoding algorithm which decode both
  errors and erasures.  Then, we focus on codes over integer rings and
  how to decode them.  We are then faced with the problem of the
  exponential growth of intermediate values, and to circumvent the
  problem, it is natural to propose to do computations modulo a prime
  ideal.  For this, we study the reduction of generalized Gabidulin
  codes over number ideals codes modulo a prime ideal, and show they
  are classical Gabidulin codes.  As a consequence, knowing side
  information on the size of the errors or the message, we can reduce
  the decoding problem over the integer ring to a decoding problem
  over a finite field. We also give examples and timings.
  

\end{abstract}

\section{Introduction}

Gabidulin codes and rank metric, introduced in \cite{delsarte1978bilinear} from a combinatorial point of view and in \cite{gabidulin1985theory} from an algorithmic and algebraic point of view, play an important role in coding theory as well as in cryptography.  From a coding theory point of view
they are adapted to correct errors and erasures that occur, either along lines of a matrix as could be the case on chip array storage and magnetic 
tapes \cite{roth1991maximum,Blaum/McEliece:1985}, or  operating on vector spaces as  in network coding~\cite{koetter2008coding}. 
In the field of cryptography, rank metric and Gabidulin codes have been used in the design of code-based public-key cryptosystems, see for instance  \cite{Gabidulin/Paramonov/Tretjakov,faure2006new}.

The goal of this paper is to generalize the construction of rank metric and Gabidulin codes, already well established
for finite fields, to any type of fields, in particular number fields, and to study how the algebraic and 
algorithmic properties are transposed. To do this we make extensive use of Ore theory of $\theta$-polynomials that
are the natural generalization of linearized polynomials, \cite{ore1934contributions,ore1933theory}. 

In {\bf Section 2}, we introduce the ring of $\theta$-polynomials as well as a suitable evaluation operator.
Then we propose several definitions for the rank metric and provide a  framework in which all the definitions are equivalent,
generalizing faithfully the finite field case. In this framework, we define generalized Gabidulin codes as being
evaluation codes of $\theta$-polynomials of bounded degree  on a so-called support consisting of linearly independent elements. 
We show in particular that these codes are also optimal. The reader interested in the applications of $\theta$-polynomials to coding theory 
and the relations between the different types of evaluations can refer to \cite{boucher2007skew,boucher2012linear}.  

In {\bf Section 3} we deal with the decoding problem of the generalized Gabidulin codes. We show that in the case where only 
errors occur, finding an error of rank less than the error-correcting capability can be done first by solving a system of linear equations, secondly by  computing 
a Euclidean division on the left in the ring of $\theta$-polynomials. Concerning  decoding in presence of errors and erasures, we recall two known models of erasures (line erasures and network coding erasures) and show that, using linear algebra and puncturing positions, this amounts to 
decoding errors in another generalized Gabidulin code, which again reduces  to  solving of a linear system. 

An efficient way to solve this linear system, inspired by the so-called Welch-Berlekamp algorithm~\cite{berlekamp1986error,gemmell1992highly}, is presented in {\bf Section 4}. This algorithm requires a quadratic number of arithmetic operations. 
We also give two variants which enable to decrease the practical complexity.

In {\bf Section 5} we address the problem of 
controlling of the size of coefficients in infinite fields, especially in
number fields.  The decoding algorithm involves coefficients whose
size increases exponentially.  Focusing on integral codes, which are
restriction of generalized Gabidulin codes to the ring of integral
elements, we establish conditions to reduce the code modulo a
prime ideal and prove that reduction and decoding are
compatible: decoding the integral code can be done decoding the code modulo the prime ideal.

Finally we present examples in {\bf Section 6}.  The first one shows a
full run of the decoding algorithm, with all the intermediate internal
values.  The second one shows the benefits of computing in the residue
field, by providing several timings.

\section{Generalization of Gabidulin codes}
In this section, we aim to generalist Gabidulin codes to the case of
an algebraic extension $K\hookrightarrow L$ of any field $K$, in
particular infinite.  Given a $K$-automorphism $\theta\in\Auto _K(L)$,
we first define $\theta$-polynomials which are a natural
generalization of $q$-polynomials, described in~\cite{ore1933special}
and which are classically used to design Gabidulin codes in finite
fields.  We present their properties and give proofs  that these properties are
independent of the underlying field, under some hypothesis $\Hdim$.  We give several possible
definitions of the {\em rank metric} and prove that they are
equivalent under the hypothesis $L^{\theta}=K$.  
Then, the hypotheses
$\Hdim$ and $L^{\theta}=k$ are equivalent, and hold for cyclic extensions.  When these hypotheses are valid, we have  a nice framework for studying Gabidulin codes
and give their main properties.

\subsection{Ore rings and $\theta$-polynomials}\label{sssec:thpo}

A $q$-polynomial in a finite field extension
$\F_q \hookrightarrow \F_{q^m}$ is a polynomial of the form
$\sum_i a_i X^{q^i}$. In the literature related to coding theory, they
are also called linearized polynomial~\cite{Berlekamp:1968,lidl1997finite}.  When
the Frobenius automorphism $x \mapsto x^q$ is replaced by an
automorphism $x\mapsto \theta(x)$ of an extension field, these polynomials are called
$\theta$-polynomials.  Originally they were introduced by Ore in the
1930's (see~\cite{ore1933theory} for the general theory
and~\cite{ore1933special} for  $q$-polynomials).  In this
Section, we recall some useful facts and we give proofs independent of
the finiteness or not of the fields.

\begin{definition}[$\theta$-polynomials]
\label{def:thetaPoly}
Let $K \hookrightarrow L$ be a field extension of finite degree
$m=[L:K]$ and let $\theta \in \Auto_{K}(L)$ be a $K$-automorphism.  A
\emph{$\theta$-polynomial} with coefficients in $L$ is an element of
the form
	\[
		\sum_{i \geq 0} a_i \Monome{X}{i}, ~ a_i \in L,
	\] 
	with a finite number of non-zero $a_i$'s.
\end{definition}

\begin{definition}[$\theta$-degree]
\label{def:thtaDegree}
The \emph{degree} of a $\theta$-polynomial
$A(X) = \sum a_i \Monome{X}{i}$ is 
\[
\deg(A)=
\begin{cases}
 - \infty &\text{ if } A=0,\\
\max\{i : a_i \neq 0\}& \text{ if } A\neq 0.
\end{cases}
\]

\end{definition}

\begin{definition}[Ring of $\theta$-polynomials]
	\label{def:RingThetaPoly}
	We denote by $L[X;\theta]$ the \emph{set of
          $\theta$-polynomials}, provided with the following
        operations. Let $A(X) = \sum a_i \Monome{X}{i}$,
        $B(X) = \sum b_i\Monome{X}{i}$ $\in L[X;\theta]$ and $c \in L$:
	\begin{itemize}
        \item the \emph{addition} is defined component-wise:
          $A(X)+B(X) = \sum_i (a_i+b_i) \Monome{X}{i}$;
        \item the \emph{(symbolic) product} is defined by
          $X \cdot c = \theta(c) X$ and $X^i\cdot X^j=X^{i+j}$.
	\end{itemize}
	The \emph{product} of $A(X)$ and $B(X)$  is then given by
        $A(X) \cdot B(X) = \sum_{i,j} a_i \AppliqueTheta{b_j}{i}
        \Monome{X}{i+j}$.
\end{definition}

\begin{remark}
  This product is called the \emph{symbolic product}
  (see~\cite{ore1933special}) in order to make a distinction from the
  classical product in polynomial ring over finite fields.
\end{remark}
Then $L[X;\theta]$ is a non commutative ring, whose unit element is
$\Monome{X}{0}$.  It admits no zero divisors and is a left and right
Euclidean ring~\cite{ore1933theory}.  A major difference with the
polynomial ring $L[X]$ occurs when defining the evaluation of
$\theta$-polynomials on scalars. In the usual case, the evaluation of
$A(X) \in L[X]$ at $b\in L$, is simply defined by ``replace $X$ by
$b$'', which is equivalent to computing the remainder of the Euclidean
division of $A(X)$ by $X-b$.  This equivalence does not hold for
$\theta$-polynomials.

Concerning $\theta$-polynomials, there are  two different types of
evaluations. The first one consists in taking the remainder of the
right Euclidean division of $A(X) \in L[X;\theta]$ by the
$\theta$-polynomial $X-b$, see \cite{boucher2012linear} for instance.
The other type is to consider the evaluation through the use of the
automorphism $\theta$ of the field $L$.  This is the evaluation
that we consider in the sequel of the paper.
\begin{definition}[Evaluation]
	\label{def:OperatorEvaluation}
	Let $A(X) = \sum a_i \Monome{X}{i} \in L[X;\theta]$ and $b\in
        L$.  The \emph{(operator) evaluation} of $A(X)$ at $b$ is
        defined by:
        \[
          \Eval{A}{b} = \sum_i a_i \AppliqueTheta{b}{i} 
        \]
\end{definition}
The properties of this evaluation are as follows.
\begin{proposition}
	\label{prop:OperatorEvaluation}
	Let $K \hookrightarrow L$, $\theta \in \Auto_{K}(L)$ and let
        $A(X)$, $B(X)$ $\in L[X;\theta]$, $a$, $b$ $\in L$ and $\lambda,\mu$
        $\in K$.  Then we have
	\[
        \begin{array}{rcl}
	\Eval{A}{\lambda a + \mu b} &=& \lambda \Eval{A}{a} +\mu \Eval{A}{b}, \\
	\Eval{(AB)}{a} &=& \Eval{A}{\Eval{B}{a}}. 
        \end{array}
	\]
\end{proposition}
The roots of a $\theta$-polynomial are then naturally defined as follows.

\begin{definition}[Roots and root-space]
	\label{def:ThetaRoots}
	Let $A(X) \in L[X;\theta]$. An element $b \in L$ is a \textit{root} of $A(X)$ if $\Eval{A}{b}=0$.
	We define the \textit{root-space} of $A(X)$ by 
	\[\roots(A(X)) = \{ b \in L : \Eval{A}{b}=0 \}. \]                
\end{definition}
From \textbf{Proposition~\ref{prop:OperatorEvaluation}}, it is obvious
that $\roots(A(X))$ is a $K$-vector space.  For a classical polynomial
in $L[X]$, the number of roots is upper bounded by the degree.  In the
case of $\theta$-polynomials, we are looking for a similar property.
\begin{definition}[Hypothesis $\Hdim$]\label{df:Hdim}
  Let $K \hookrightarrow L$ be a field extension and
  $\theta \in \Auto_{K}(L)$.  We say that $L[X;\theta]$ verifies the
  \emph{hypothesis $\Hdim$} if for all non-zero $\theta$-polynomial
  $A(X)$ we have
	\[ \dim_K  \roots (A(X)) \leq \deg A(X). \]
\end{definition}
We have the following immediate proposition, which paves the way for a Lagrange interpolation.
\begin{proposition}
	\label{crl:hdim}
	Suppose that $L[X;\theta]$ has the property $\Hdim$. Let
        $v_1, \ldots, v_s \in L$ be $K$-linearly independent, and
        $A(X) \in L[X;\theta]$ such that
        $\Eval{A}{v_i}=0, i=1, \ldots, s$.  Then either $A(X)=0$, or
        $\deg A(X) \geq s$.
\end{proposition}
 For a given vector space, there exists a $\theta$-polynomial
vanishing on the vector space.
\begin{theorem}
	\label{thm:degann}
	Let $K \hookrightarrow L$ be a field extension of finite
        degree $m$, and $\theta \in \Auto_{K}(L)$.  Let $V \subset L$ be a
        $K$-linear subspace of dimension $n$.  There exists a unique
        monic polynomial $A(X)\in L[X;\theta]$ which vanishes on all
        $x\in V$. If furthermore $L[X;\theta]$ verifies $\Hdim$, $A(X)$ has degree exactly $n$.
\end{theorem}
\begin{proof}
  We first prove existence and unicity. Since $[L:K]$ is finite,
  $\theta$ has finite order, say $s$, and we have
  $\Eval{(\Monome{X}{s}-\Monome{X}{0})}{x}=0$, for all $ x \in L$. In
  particular, $\Eval{(\Monome{X}{s}-\Monome{X}{0})}{x}=0$, for all
  $x\in V$, and there exists a non zero polynomial vanishing on $V$.
  Let $A(X)$ and $B(X)$ be two non zero monic polynomials both of
  minimal degree and vanishing on $V$.  They have the same degree and
  they are both monic.  Then, for all $ v \in V$,
  $\Eval{(A(X)-B(X))}{v}=0$, but $A(X)-B(X)$ has lower degree than
  $A(X)$ or $B(X)$, which contradicts the minimality of the
  degree. Thus $A(X)=B(X)$.

  Assume that $L[X;\theta]$ verifies $\Hdim$.  By {\bf
    Proposition~\ref{crl:hdim}}, the degree of any non-zero
  $\theta$-polynomial annihilating $V$ is at least $n$. Now we give
  an explicit construction of this polynomial: let
  $(v_1, \cdots, v_s)$ be a basis of $V$ and let the sequence
  $A_i(X) \in L[X; \theta]$, $i=0,\dots,n$, be recursively defined by
    \begin{eqnarray}
      A_0(X) &= &1, \\
      A_i(X) &= &\left( \Monome{X}{1} - \frac{\theta(\Eval{A_{i-1}}{v_i})}{\Eval{A_{i-1}}{v_i}} 
      \Monome{X}{00} \right) \cdot A_{i-1}(X)\label{eq:rec:annihil}.			
    \end{eqnarray}
  By induction it is easy to show that, for $i=1,\ldots,n$, $A_i(X)$
  is monic of degree $i$ and that for all $1 \leq j \leq i$,
  $\Eval{A_{i}}{v_j} = 0$.  This also proves that the quotient in~(\ref{eq:rec:annihil}) above is
  well-defined, since $\Eval{A_{i-1}}{v_i} \neq 0$ by $\Hdim$.
  Therefore, $A(X) = A_n(X)$ is monic of degree $n$ and its 
    root-space is $V$.        \qed\end{proof} 
\begin{definition}[Annihilator polynomial]\label{df:annil}
  Let $K \hookrightarrow L$ be a field extension, and
  $\theta \in \Auto_{K}(L)$.  Let $V \subset L$ be a $K$-linear
  subspace of dimension $n$.  The \emph{annihilator polynomial} of $V$
  is the monic $\theta$-polynomial $\Annul{V}(X) \in L[X; \theta]$ of
  minimal degree such that
  $ \text{for all } v \in V, \Eval{\Annul{V}}{v}=0$.  Given a family
  $\vectgr g=(g_1,\dots,g_n)\in L^n$, we will use the notation
  $\Annul{\vectgr g}(X)$ for the annihilator polynomial $\Annul{V}$,
  where $V$ is the $K$-linear space spanned by $\vectgr g$.
\end{definition}
We can also define interpolating polynomials, and when $\Hdim$ holds, interpolating polynomials have the natural expected
degree.
\begin{theorem}\label{the:int}
  Let $K \hookrightarrow L$ be a field extension of finite degree
  $[L:K] = m$ and $\theta \in \Auto_{K}(L)$ such that $L[X;\theta]$
  verifies $\Hdim$.  Let $\vectgr g=(g_1, \dots, g_n)$ be $K$-linearly
  independent elements in $L$ and $\vectgr y=(y_1, \dots, y_n)\in L^n$.  There
  is a unique $\theta$-polynomial of degree at most $n-1$ such that
	\[ 
		 \Eval{\Int}{g_i} = y_i,~ 1 \leq i \leq n. \label{eq:interp}
	\]
\end{theorem}

\begin{proof}
  For $i=1,\dots,n$, let
  $\vectgr {\widehat{g}}_{i}\stackrel{def}{=}(g_1, \ldots, g_{i-1},g_{i+1},
  \ldots, g_n)$, and consider
     \[
	 	\Int(X) \stackrel{def}{=} \sum_{i=1}^{n}{ y_i \frac
		{  		 \Annul{  \vectgr {\widehat{ g}_{i}  }}  (X)   }
		{  \Eval{  \Annul{ \vectgr{ \widehat{ g}_{i} }}  }{ g_i }  }} \ .
	\]
        It is easy to check that $ \Int(X)$ satisfies the conditions
        of {\bf Theorem \ref{the:int}}.  Suppose now that there exists
        $B(X) \in L[X;\theta]$ of degree $n-1$ such that
        $\Eval{B}{g_i} = y_i \ForallInt{1}{i}{n}$.  Then
        $\Int(X) - B(X)$ has degree $\leq n-1$ and
        $\Eval{(\Int-B)}{g_i} = 0$ for $1 \leq i \leq n$.  The
        hypothesis $\Hdim$ implies that that $B(X) = \Int(X)$ and
        we have unicity.  \qed\end{proof}

\begin{definition}[Interpolating polynomial]\label{def:interp}
  The polynomial introduced at {\bf Theorem \ref{the:int}} is called
  the \emph{interpolating $\theta$-polynomial} of $\vectgr y$ at
  $\vectgr g$, and is denoted by $\Interpol{\vectgr g}{\vectgr y}$.
\end{definition}


\subsection{Rank metrics}\label{sssec:rang}

Over finite field extensions, the rank metric was studied by
Delsarte~\cite{delsarte1978bilinear} from a combinatorial point of
view.  Since in this paper our aim is to extend the notion of rank
metric to infinite fields, the definitions introduced here are
compatible with the finite field case.

Given a field $K$, we let $\mathcal{M}_{r,c}(K)$ denote the ring of
matrices with $r$ rows, $c$ columns and coefficients in~$K$.  Given
$ K \hookrightarrow L$ a field extension, and a matrix $M$ with
coefficients in $L$, we write $\rang_L(M)$ for the maximal number of
linearly independent columns over the field $L$, which is the usual
rank of a matrix, and $\rang_K(M)$ for the maximal number of linearly
independent columns over the field $K$.

\begin{definition}[Four rank metrics]\label{def:rang}
  Let $ K \hookrightarrow L$ be a field extension of finite degree
  $[L:K]=m$, and let $\theta \in \Auto_{K}(L)$ be an automorphism of
  order $s$.  Let $\mathcal{B}=\VecteurLigneCoins{b_1}{b_m}$ be a
  $K$-basis of $L$, and for
  $\vectgr{x}=\VecteurLigneCoins{x_1}{x_n} \in L^n$, 
  let $x_{ij}\in K$, $i=1,\dots,n$, $j=1,\dots,m$, be the coordinates of the $x_i$'s in the $K$-basis
  $\mathcal B$ as follows: 
\[
x_i=\sum_{j=1}^mx_{ij}b_j, \quad
  i=1,\dots,n.
\]
Considering
	\[
		\matrice{B_{x,\mathcal B}}=\MatriceGeneraleEch{x}{m}{n} \in \M_{m \times n}(K),
	\]
and 
	\[
          \matrice{V_{x,\theta}}=\MatriceCoins{x_1}{x_n}{\AppliqueTheta{x_1}{s-1}}{\AppliqueTheta{x_n}{s-1}}
          \in \M_{s \times n}(L),
	\]
        we define the four following \emph{weights} for $\vectgr{x}\in L^n$:
\begin{align*}
  \wa(\vectgr{x}) 	&=\deg\Annul{\vectgr x}(X),\\
	\wthL(\vectgr{x}) 	&= \rang_L \matrice{V_{x,{\theta}}},  \\
	\wthK(\vectgr{x}) 	&= \rang_K \matrice{V_{x,{\theta}}},  \\
	\wb(\vectgr{x}) 	&= \rang_K \matrice{B_{x,{\mathcal{B}}}}. 
\end{align*}
\end{definition}
\begin{definition}[Rank distances]
	We endow $L^n$ with the  distances induced by previous weights:
	\[ 
\begin{array}{rcl}
  \da(\vectgr{x},\vectgr{y}) 	&=&\wa(\vectgr x-\vectgr y), \\
  \dthL(\vectgr{x},\vectgr{y}) 	&=&\wthL(\vectgr x -\vectgr y),\\
  \dthK(\vectgr{x},\vectgr{y}) 	&=&\wthK(\vectgr x-\vectgr y ),\\
  \db(\vectgr{x},\vectgr{y}) 	&=&\wb(\vectgr x-\vectgr y).\\
\end{array}
        \]
These distances are called \emph{rank distances}.
\end{definition}
It is clear that the last three  weights  induce  distance over $L^n$. We need a proof for $\da$.
\begin{proposition}
The metric $\da$ is a distance.
\end{proposition}
\begin{proof}
  We prove the associated relevant statements for the weight function $\wa$. Let
  $\vectgr x\in L^n$ such that $\wa(\vectgr x)=0$. This means that the
  annihilator polynomial of $\vectgr x$ has degree 0, and, being monic, it is the polynomial $A(X)=1\in L[X,\theta]$. Now, for
  $i=1,\dots,n$, we have $A\{x_i\}=1\{x_i\}=x_i=0$. Thus
  $\wa(\vectgr x)=0$ implies $\vectgr x=0$. Now we have easily that
  $\wa(-\vectgr x)=w(\vectgr x)$, since, for any $A(X)\in L[X;\theta]$
  which vanishes on $\vectgr x$, we have $A\{-x_i\}=-A\{x_i\}=0$,
  $i=1,\dots,n$, and the annihilator polynomials of $\vectgr x$ and
  $-\vectgr x$ are the same.

  We finally have to prove that, given
  $\vectgr x_1,\vectgr x_2\in L^n$,
  $\wa(\vectgr x_1+\vectgr x_2)\leq \wa(\vectgr x_1)+\wa(\vectgr x_2)$. 
  Let $\wa(\vectgr x_1)=w_1$ and $\wa(\vectgr x_2)=w_2$, and
  consider the following elements in $L^n$, where addition and $\theta$ are applied component-wise:
  \begin{align*} 
    \vectgr y_0&= \vectgr x_1 + \vectgr x_2,\\ 
    \vectgr    y_1&=\theta( \vectgr x_1 + \vectgr x_2)=\theta(\vectgr x_1)+\theta(\vectgr x_2),\\ 
               &\vdots\\ 
    \vectgr    y_{w_1+w_2}&= \theta^{w_1+w_2}(\vectgr x_1 + \vectgr x_2)
                          =\theta(\vectgr x_1)^{w_1+w_2}+\theta(\vectgr x_2)^{w_1+w_2}.
  \end{align*} Now using that $\theta^i(\vectgr x_1+\vectgr
x_2)=\theta^i(\vectgr x_1)+\theta^i(\vectgr x_2)$, $i=1\dots,w_1+w_2$,
and considering the rest of the right division of $X^i $ by
$\Annul{\vectgr x_1}(X)$ and by
$\Annul{\vectgr x_2}(X)$, we see that $\theta^i(\vectgr
x_1)+\theta^i(\vectgr x_2)$ belong the $L$-vector space
  \[ V_{\vectgr x_1,\vectgr x_2}=\Vect_L\left\langle\vectgr x_1,\theta
(\vectgr x_1),\dots,\theta(\vectgr x_1)^{w_1-1},\vectgr
x_2,\theta(\vectgr x_2),\dots,\theta^{w_2-1}(\vectgr
x_2)\right\rangle.
  \] Now, since $\dim_LV_{\vectgr x_1,\vectgr x_2}\leq w_1+w_2$, there
exists an $L$-linear non zero dependency between the $w_1+w_2+1$
vectors $\vectgr y_0,\dots,\vectgr y_{w_1+w_2}$. Then, by {\bf Proposition~\ref{crl:hdim}}, 
the annihilator polynomial of
  $\vectgr x_1+\vectgr x_2$ has degree less than or
  equal to $w_1+w_2$.\qed\end{proof}
The natural question about the relationship between the different induced metrics is answered in the following theorems.
\begin{theorem}\label{thm:equivrangL}
	For all  $\vectgr{x} \in L^n$,
$
          \wa(\vectgr{x}) = \wthL(\vectgr{x})$.
\end{theorem}

\begin{proof} 
  The case of $\vectgr x=0$ being trivial, let us consider
  $\vectgr{x}=\VecteurLigneCoins{x_1}{x_n} \neq 0 \in L^n$.  We first
  show that $\wthL(\vectgr{x}) \leq \wa(\vectgr{x})$.  Let
  $w=\wa(\vectgr{x})$.  Let
  $\Ann_{\vectgr{x}}(X) = \sum_{i=0}^{w} a_i \Monome{X}{i}$ be the
  monic annihilator polynomial of $\vectgr x$.  Then, let $L_i$,
  $i=0,\ldots,s-1$ be the $i+1$-th row of $\matrice{V_{x,\theta}}$. We
  have $L_{w}=-\sum_{i=0}^{w-1}{a_{i-1} L_{i}}$ and the $(w+1)$-th row
  is a linear combination of the previous ones.  Then, applying
  recursively $\theta$ on this relation, we can express any row of
  index larger than or equal to $w+1$ as a linear combination of the
  $w$ first rows.  Thus, $\matrice{V_{x,{\theta}}}$ has at most
  $\wa(\vectgr{x})$ linearly $L$-independent rows which is the same as
  the number of $L$-independent columns. Thus $ \wthL(\vectgr{x})\leq \wa(\vectgr x)$.

  Conversely, let $w=\wthL(\vectgr{x})$. Let $1 \le u \le w $ be the
  smallest index such that $L_1,\ldots,L_u$ are $K$-Linearly
  independent. Since, by construction of $\matrice{V_{x,\theta}}$,
  $L_{i+1} = \theta(L_i)$ where $\theta$ acts on all the components of
  the row $L_i$, any row $L_i$, $i \ge u+1$, is a linear combination
  of the rows $L_1,\ldots,L_u$. Therefore, $u=w$, and there is a
  linear combination such that
  $L_{w} + \sum_{i=0}^{w-1}{\mu_i L_i}=0$, and the polynomial
  $X^w + \sum_{i=0}^{w-1}{\mu_iX^i}$ annihilates $\vectgr{x}$. Hence
  $ \wa(\vectgr{x}) \le w=\wthL(\vectgr{x})$.  \qed\end{proof} For the
next inequalities, we have to compare the number of linearly
independent columns of two matrices.  For that, we will consider a
linear combination of columns of the first one, and show that it gives
a linear combination in the second matrix.  Since the order of the
columns, reflecting the ordering of the $x_i$'s in $\vectgr x$, does
not matter, w.l.o.g.\ we can order them such that the $r$ first
columns are linearly independent, where $r$ denotes the rank of the
matrix.
\begin{theorem}\label{thm:equivrangK}
	For all  $\vectgr{x} \in L^n$, $\wthK(\vectgr{x}) = \wb(\vectgr{x})$.
\end{theorem}

\begin{proof}
  The case of $0$ being trivial, let us consider
  $\vectgr{x}=\VecteurLigneCoins{x_1}{x_n} \neq 0 \in L^n$.  We show
  that $\wthK(\vectgr x) \leq \wb(\vectgr x)$.  Let
  $w=\wb(\vectgr{x})=\rang_K (\matrice{B_{x,\mathcal{B}}})$, and
  w.l.o.g.\ assume that the first $w$ columns of $\matrice{B_{x,\mathcal{B}}}$ are linearly
  independent. Consider the $u$-th column expressed as a linear
  combination  of the first $w$  columns of $\matrice{B_{\vectgr x,\mathcal{B}}}$:
	\[
        \VecteurCoins{x_{u,1}}{x_{u,m}} = \lambda_1 \VecteurCoins{x_{1,1}}{x_{1,m}} + \cdots + \lambda_w \VecteurCoins{x_{w,1}}{x_{w,m}}, 
        \] 
        where $\lambda_i \in K$ for $i=1,\ldots, w$. For $i=1,\dots,n$, writing
        \[
          x_i=\sum_{j=1}^m x_{ij} b_j,
        \]
        we have
	\[
        x_u = \lambda_1 x_1 + \cdots + \lambda_w x_w. 
        \]
	By applying the $K$-automorphism $\theta^k$, for $k=1,\dots,s$,
        we get
	\[
        \AppliqueTheta{x_u}{k} = \lambda_1 \AppliqueTheta{x_1}{k} + \cdots + \lambda_w \AppliqueTheta{x_w}{k}.
        \]
	Hence, 
	\[ 
		\VecteurCoins{x_u}{\AppliqueTheta{x_u}{s-1}} = 
		\lambda_1 \VecteurCoins{x_1}{\AppliqueTheta{x_1}{s-1}} + \cdots 
		+ \lambda_w \VecteurCoins{x_w}{\AppliqueTheta{x_w}{s-1}}, 
	\]
	and a linear combination in $\matrice{B_{x,{\mathcal{B}}}}$
        implies a linear combination in $\matrice{V_{x,{\theta}}}$, thus
        the $K$-rank of $\matrice{V_{x,{\theta}}}$ is at most the $K$-rank of
         $\matrice{B_{x,{\mathcal{B}}}}$.

         Now, we show that $\wb(\vectgr{x}) \leq \wthK(x)$.  Let
         $w=\wthK(\vectgr{x})=\rang_K(\matrice{V_{x,{\theta}}})$.
         Consider the expression of the $u$-th column as a $K$-linear
         combination of the $w$ first columns of
         $\matrice{V_{x,\theta}}$:
	\[
        \VecteurCoins{x_u}{\AppliqueTheta{x_u}{s-1}} = \lambda_1 \VecteurCoins{x_1}{\AppliqueTheta{x_1}{s-1}} + \cdots + \lambda_w \VecteurCoins{x_w}{\AppliqueTheta{x_w}{s-1}} ,
        \]
	where $\lambda_i \in K$, for all $i=1,\ldots,w$.  Then, by
        considering the first row, we have
        $ x_u = \lambda_1 x_1 + \cdots + \lambda_w x_w$, which can be
        expanded in the basis $\mathcal{B}$ as follows:
	\[
        x_{u,j} = \sum_{i=1}^w \lambda_i x_{i,j}, 1 \leq j \leq m .
        \]
	Thus,
	\[ \VecteurCoins{x_{u,1}}{x_{u,m}} = \lambda_1 \VecteurCoins{x_{1,1}}{x_{1,m}} + \cdots + \lambda_w \VecteurCoins{x_{w,1}}{x_{w,m}} 
         \]
	so $\wb(\vectgr{x}) = \rang_K (X_{\mathcal{B}}) \leq w$.
\qed\end{proof}
\begin{proposition} We have the following inequality, for all $\vectgr x\in L^n$:
\[
 \wa(\vectgr{x}) = \wthL(\vectgr{x}) \leq  \wthK(\vectgr{x}) = \wb(\vectgr{x}).
\]
\end{proposition}
\begin{proof}
	A linear combination with coefficients in $K$ is also a linear combination with coefficients in $L$, so $\wthL(\vectgr{x}) \leq \wthK(\vectgr{x})$.

\qed\end{proof}
Now we introduce a condition for all these weights to be equal.

\begin{proposition}\label{prop:egaliterang}
	If  $L^{\theta} = K$, where $L^{\theta}=\{ x \in L : \theta(x)=x \}$ is the fixed field of $\theta$, then, for all $     \vectgr{x} \in L^n$,
\[
\wa(\vectgr{x})  = \wthL(\vectgr{x}) = 	\wthK(\vectgr{x})  =  \wb(\vectgr{x}).
\]
\end{proposition}

\begin{proof}
  From previous Theorems, it suffices to show that for all
  $\vectgr{x}=\VecteurLigneCoins{x_1}{x_n} \in L^n$,
  $\wthK(\vectgr{x}) \leq \wthL(\vectgr{x})$.  Let
  $w=\wthL(\vectgr{x})=\rang_L(\matrice{V_{x,{\theta}}})$.  Consider the
  expression of the $u$-th column as an $L$-linear combination of
  of the $w$ first columns of $\matrice{V_{x,\theta}}$ which form a
  basis of the column space of  $\matrice{V_{x,\theta}}$:
\begin{equation}\label{Eq:LinCombi}
\VecteurCoins{x_n}{\AppliqueTheta{x_n}{s-1}} = \mu_1 \VecteurCoins{x_1}{\AppliqueTheta{x_1}{s-1}} + \cdots + \mu_w \VecteurCoins{x_w}{\AppliqueTheta{x_w}{s-1}}, 
	\end{equation}
	with $\mu_i \in L$, $i=1,\dots,w$. By applying $\theta$, we obtain 
	\[ 
          \VecteurCoins{\AppliqueTheta{x_u}{1}}{\AppliqueTheta{x_u}{s}} 
          = \AppliqueTheta{\mu_1}{} \VecteurCoins{\AppliqueTheta{x_1}{1}}{\AppliqueTheta{x_1}{s}} 
          + \cdots 
          + \AppliqueTheta{\mu_w}{} \VecteurCoins{\AppliqueTheta{x_w}{1}}{\AppliqueTheta{x_w}{s}}.  	
	\]
	Since $s$ is the order of $\theta$, i.e.\ $\theta^s=Id$, by
        reordering the lines, we obtain
	\[
		\VecteurCoins{x_n}{\AppliqueTheta{x_n}{s-1}} 
		= \AppliqueTheta{\mu_1}{} \VecteurCoins{x_1}{\AppliqueTheta{x_1}{s-1}} 
		+ \cdots 
		+  \AppliqueTheta{\mu_w}{} \VecteurCoins{x_w}{\AppliqueTheta{x_w}{s-1}}. 
	\] 
	Since the first $w$ columns form a basis of the column space,
        the decomposition (\ref{Eq:LinCombi}) is unique. Therefore,
        $\mu_i=\theta(\mu_i)$, i.e. $\mu _i \in L^{\theta} =
        K$, for 
        $i=1,\ldots,n$. This implies that
        $\wthK(\vectgr{x})= \rang_K(\matrice{V_{x,{\theta}}}) \leq
        \wthL(\vectgr{x}) = \rang_L(\matrice{V_{x,{\theta}}})$.
        \qed\end{proof}
When $L^{\theta}=K$, the four metrics previously defined are identical. 
This metric is called \textit{the} rank metric, and denoted by $w(\vectgr{x})$.


\subsection{A framework for  studying  Gabidulin codes}\label{sssec:exfr}

From previous Sections, we have seen that when $L[X;\theta]$ verifies
$\Hdim$, the dimension of the root-space of a $\theta$-polynomial
$A(X)$ is at most equal to the degree of $A(X)$. Also, when
$L^{\theta}$, the field fixed by $\theta$, is $K$, we have that all
the defined metrics are equal. We first present two examples, of a bad
situation and of a nice situation.

\begin{example}
\label{Ex:KulExt1}
Let  $K = \Q[j] = \Q[X]/(X^2+X+1)$, which  is a field, which contains 
the sixth roots of the unity $\{1, j, j^2, -1, -j, -j^2 \}$ where $j^3=1$. 
We build an extension $L$ of $K$ by adding a sixth root of $2$, denoted by $\alpha$, and $\alpha^6=2$.
Since square and cubic roots of $2$ are not in $K$, we get the following  \emph{Kummer extension}~\cite[§IV.3]{neukirch1999algebraic}:
\[
 L = K[\alpha] = K[Y]/(Y^6-2). 
\]
Hence $[L:K] = 6$ and any $K$-automorphism of $L$ is uniquely defined
by the image of $\alpha$, which must be a root of $Y^6-2$.  Let us
consider $\theta_1 : \alpha \mapsto j \alpha$, which  has order
$3$: $\theta_1^3=Id_L$.
Then $\Hdim$ is not verified, as can be seen by considering the
$\theta_1$-polynomial $A=\Monome{X}{1}-\Monome{X}{0}$.
One can check that $\roots(\Monome{X}{1}-\Monome{X}{0})$ contains $1$ and $\alpha^3$,
therefore has $K$-dimension $2$, which is twice the degree of $A$.
We also have that the fixed field  $L^{\theta_1}$ is spanned by $1$ and $\alpha^3$. Therefore  $K \ne L^{\theta_1}$. The rank metrics are clearly different:   
\[
\wa(1, \alpha, \alpha^3, \alpha^4)=2 ~< ~\wb(1, \alpha, \alpha^3, \alpha^4)=4. 
\]
Indeed, we have that
$j\Monome{X}{00}-(j+1)\Monome{X}{1}+\Monome{X}{2}$ vanishes on
$(1, \alpha, \alpha^3, \alpha^4)$, and thus
\[ \wa(1, \alpha, \alpha^3, \alpha^4)=\Deg \left(j\Monome{X}{00}-(j+1)\Monome{X}{1}+\Monome{X}{2} \right) =2; \]
while
\[ \wb(1, \alpha, \alpha^3, \alpha^4)=\rang \begin{pmatrix}
												1&0&0&0 \\ 0&1&0&0\\ 0&0&0&0\\ 0&0&1&0\\0&0&0&1\\0&0&0&0
											\end{pmatrix} = 4.\]


\end{example}

\begin{example} We consider the field extension of
  Example~\ref{Ex:KulExt1}, and now consider
  $\theta_2 : \alpha \mapsto (j+1) \alpha$.  The automorphism
  $\theta_2$ has order $6$.  We have that
  $\roots(\Monome{X}{1}-\Monome{X}{0})$ has dimension $1$ and
  $L^{\theta_2} = K$, which implies the equality of the metrics, for example:
\[
w_0(1, \alpha, \alpha^3, \alpha^4)=w_3(1, \alpha, \alpha^3, \alpha^4)=4.
\]
Indeed, we have: 
\[ \wa(1, \alpha, \alpha^3, \alpha^4)=\Deg \left(j\Monome{X}{00}-(j+1)\Monome{X}{2}+\Monome{X}{4} \right) =4; \]
and
\[ \wb(1, \alpha, \alpha^3, \alpha^4)=\rang \begin{pmatrix}
												1&0&0&0 \\ 0&1&0&0\\ 0&0&0&0\\ 0&0&1&0\\0&0&0&1\\0&0&0&0
											\end{pmatrix} = 4.\]
In this case, $\Hdim$ is verified.  
\end{example}
These examples show that there is a connection 
between the hypothesis $\Hdim$ and $L^{\theta}=K$.  The following
theorem establishes this link.
\begin{theorem}\label{thm:framework}
	Let $ K \hookrightarrow L$ be a field extension of finite degree $[L:K]=m$, 
	and let $\theta \in \Auto_{K}(L)$.
	The following statements are equivalent:
	\begin{enumerate}[label=\roman*]
		\item the subfield $L^{\theta}$ is  $K$,
		\item the weights $\wa$ and $\wb$ are equal,
		\item the annihilator polynomial of $n$ $K$-linearly independent elements of $L$ has degree exactly $n$,
		\item the dimension of the {\em root-space} of any non-zero polynomial is upper-bounded by its degree.
	\end{enumerate}
\end{theorem}
\begin{proof} \hfill
	\begin{itemize}
\item[$i \Rightarrow ii$] This implication has been proved in {\bf Proposition~\ref{prop:egaliterang}}.
\item[$ii \Rightarrow iii$] Let $\vectgr{v}=(v_1, \ldots, v_n)$ be a
  vector of $K$-linearly elements of $K$, meaning that
  $\wb(\vectgr{v})=n$.  The equivalence is a direct consequence of the
  definition of $\wa$.
\item[$iii \Rightarrow iv$] Let $P(X)$ be a polynomial of degree $n$ and
  $(v_1, \ldots, v_u)$ be a basis of its roots space.  Let $\Ann(X)$
  denotes the annihilator polynomial of the $v_i$'s. By $iii$, its
  degree is $u$.  We compute the Euclidean division: $P(X)=Q(X)\Ann(X)+R(X)$,
  where $\deg R(X) < \deg \Ann(X)$.  We get that $\Eval{R}{v_i}=0$, but
  since the annihilator has minimal degree, we deduce that $R(X)=0$.
  Thus, $P(X)$ is a multiple of $\Ann(X)$, and $u\leq n$.
\item[$iv \Rightarrow i$] $L^{\theta}$ is the roots space of
  $\Monome{X}{1}-\Monome{X}{0}$,  applying $\Hdim$  to this polynomial gives that $L^{\theta}$ has dimension $1$ and contains $K$, so it is $K$.\qed
\end{itemize}
\end{proof}
Since  $\theta$ has a finite order, Artin's Lemma establishes that $L^{\theta} \hookrightarrow L$
is a Galois extension of Galois group $\langle\theta\rangle$. Therefore, if the condition {\em i.} above ($K = L^{\theta}$)  is satisfied, then
 $K \hookrightarrow L$ is a Galois extension with Galois group  $\Auto_K(L)= \langle\theta\rangle$. 
Thus, a proper framework to design codes in rank metric and the generalization of Gabidulin codes is to consider 
a cyclic Galois extension $K \hookrightarrow L$. 

Finite fields extensions provided with the Frobenius automorphism are
examples of such extensions.  Concerning number fields, cyclotomic and
Kummer extensions are simple cyclic Galois extension, with explicit
generators of the Galois group.  Similarly, concerning function
fields, Kummer and Artin-Schreier extensions also provide cyclic Galois
extensions which are easy to deal with.


\subsection{Generalized Gabidulin codes}\label{sssec:ggcdef}

Gabidulin codes were defined in 1985 by Gabidulin for finite fields
\cite{gabidulin1985theory}.  They consist in the evaluation of a
$q$-polynomial of bounded degree $k$ at $n$ values.  They are the
analogue of Reed-Solomon codes in the rank metric.  In this part, we
give the Singleton bound for codes in the  rank metric, and we
generalized the construction of Gabidulin codes.  Then, we give basic
properties of these codes, with proofs that does not rely on the
finiteness of the field.

We consider $K \hookrightarrow L$ a cyclic Galois extension of degree
$[L:K]=m$ and $\langle \theta \rangle =\Auto_{K}(L)$. From previous
considerations, we endow the vector space $L^n$ with {\em the rank
  metric} $w$, using any of the four equivalent definitions.  A
$[n,k,d]_r$ code $\mathcal{C}$ is an $L$-linear subspace of dimension
$k$ of the vector space $L^n$ with minimum rank distance $d$, {\em
  i.e.}
\[
d = \min_{\vectgr{x}\in \mathcal{C}\setminus 0 } w(\vectgr{x}).
\]
With these parameters, we have 
\begin{proposition}[Singleton Bound]\label{prop:Singleton}
Let $\mathcal{C}$ be a $[n,k,d]_r$-code over $L$. Then
$        d \leq n-k+1$.
\end{proposition}

\begin{proof}
Consider the projection map $\pi$ on the first $n-d+1$ coordinates:
\[
\pi:\begin{array}[t]{rcl}
\mathcal{C}&\rightarrow& L^{n-d+1}\\
(x_1,\dots,x_n)&\mapsto &(x_1,\dots,x_{n-d+1})\end{array}
\]
and let
$ \vectgr{x}\in \mathcal C$ be such that $\pi(\vectgr{x})=0$, i.e.\
$\vectgr{x}=(0,\dots,0,x_{n-d+2},\dots,x_n)$. The number of non zero
coordinates of $\vectgr{x}$ is less than or equal to $n-(n-d+2)+1=d-1$. Then
$\Annul{\vectgr{x}}(X)$, the annihilator polynomial of $x$, right divides the left lowest common
multiple of
\[
X-\frac{\theta(x_{n-d+2})}{x_{n-d+2}},\dots,X-\frac{\theta(x_{n})}{x_{n}}
\]
which has degree less than or equal to $d-1$. Thus $w(\vectgr{x})\leq d-1$:
since the minimum distance of $\mathcal C$ is $d$, $x=0$, and
$\pi$ is injective. As a consequence $k\leq n-d+1$.\qed
\end{proof}

\begin{definition}[MRD codes]\label{defi:MRDCodes}
	When $d = n-k+1$, the code $\mathcal{C}$ is a \emph{MRD (Maximum Rank Distance) code}.
\end{definition}
Now we present the generalization of Gabidulin codes. 

\begin{definition}[Generalized Gabidulin code]\label{Def:GabiGen}
  Let $K \hookrightarrow L$ be a cyclic Galois extension of degree
  $[L:K]=m$ and $\Auto_K(L)=\langle \theta \rangle$.  Let
  $k \leq n \leq m$ be integers and
  $\vectgr{g}=\VecteurLigneCoins{g_1}{g_n}$ be a vector of
  $K$-linearly independent elements of $L$.  The vector $\vectgr{g}$
  is called the \emph{support} of the code.  The \emph{generalized Gabidulin
    code} $\Gab_{\theta,k}(\vectgr{g})$ is
		\begin{align} \Gab_{\theta,k}(\vectgr{g}) =\left \{ \VecteurLigneCoins{\Eval{f}{g_1}}{\Eval{f}{g_n}}~:~
f(X) \in L[X;\theta], \deg f(X)<k \right\}. \end{align}  
\end{definition}

\begin{remark}
	The case $k=n$ is useful for the theory.
	This case will appear for decoding algorithms when there is the maximal number of erasures.
\end{remark}
A codeword $\vectgr{c}=(c_1,\dots,c_n)$ of $\Gab_{\theta,k}(\vectgr{g})$ has
coefficient in $L^n$. The generalized Gabidulin code $\Gab_{\theta,k}(\vectgr{g})$
has the following generating matrix
\begin{equation}
  \label{eq:MatGenGab}
  \MatriceCoins{\AppliqueTheta{g_1}{0}}  {\AppliqueTheta{g_n}{0}}
  {\AppliqueTheta{g_1}{k-1}}{\AppliqueTheta{g_n}{k-1}}          .      
\end{equation}
Using a $K$-basis $\mathcal B=(b_1,\dots,b_m)$ of $L$, writing
$c_j=\sum_{i=1}^mc_{ij}b_j$, $c_{ij}\in K$, $i=1\dots m$,
$j=1\dots,n$, the codeword $\vectgr{c}$ can alternatively be written in matrix form:
\[
\vectgr{c}=\MatriceCoins{c_{11}}{c_{1n}}{c_{m1}}{c_{mn}}\in \M_{m \times n}(K).
\]


\begin{example} \label{ex:filrouge} We consider the cyclotomic
  extension
  $\Q \hookrightarrow \Q[X]/(1+X+\cdots+X^{6}) = \Q[\alpha] $, where
  $1+\alpha+\cdots + \alpha^{6}=0$, and the automorphism defined by
  $\theta : \alpha \mapsto \alpha^3$. In this setting, $\Q = K$ and $\Q[\alpha] =L$.
 Let
  $\mathcal{B} = (1, \alpha, \ldots, \alpha^5)$ be the considered
  $K$-basis of $L$.  Let the support be
  $\vectgr{g}=(1, \alpha, \ldots, \alpha^5)$, and we build the
  corresponding generalized Gabidulin code of length $n= 6$, dimension $k = 2$ and minimum distance $d =5$.

To encode the $\theta$-polynomial $f= \alpha^2 \Monome{X}{00} + \alpha^5 \Monome{X}{1}$, 
we compute the evaluations of $f$ on the support:
\[
	\begin{array}{rcl}
	\Eval{f}{1}			&=&			\alpha^2 + \alpha^5,	\\
	\Eval{f}{\alpha}	&=&			\alpha + \alpha^3,	\\
	\Eval{f}{\alpha^2}	&=&			2\alpha^4,			\\
	\Eval{f}{\alpha^3}	&=&			\alpha^5+1			\\
	\Eval{f}{\alpha^4}	&=&			\alpha^3+\alpha^6	\\
						&=&			-1-\alpha-\alpha^2-\alpha^4-\alpha^5\\
	\Eval{f}{\alpha^5}	&=&			\alpha^6+1			\\
						&=&			-\alpha-\alpha^2-\alpha^3-\alpha^4-\alpha^5\\	
	\end{array}
\]
Then, the corresponding codeword is the following vector:
\[
 \vectgr{c}(f) = \left(\Eval{f}{\alpha}, \ldots, \Eval{f}{\alpha^6} \right) \in \Q[\alpha]^8 ; 
\]
or the following matrix obtain by expanding the components of the vector over the basis $\mathcal{B}$:
\[
	\vecteurgras{C}(f)=	
	\begin{pmatrix}
		0 & 0 & 0 & 1 & -1 &  0 \\ 
		0 & 1 & 0 & 0 & -1 & -1 \\ 
		1 & 0 & 0 & 0 & -1 & -1 \\ 
		0 & 1 & 0 & 0 &  0 & -1 \\ 
		0 & 0 & 2 & 0 & -1 & -1 \\ 
		1 & 0 & 0 & 1 & -1 & -1
	\end{pmatrix}
	\in \mathcal{M}_{6,6}(\Q).
\]

\end{example}

\begin{proposition}\label{prop:MRD}
	Under the conditions of {\bf Definition \ref{Def:GabiGen}}, 
	$\Gab_{\theta,k}(\vectgr{g})$ is an MRD-code.
\end{proposition}
\begin{proof}
  Let $\vectgr{c}=\VecteurLigneCoins{c_1}{c_n}$ be a non-zero
  codeword\Oubliettes{ closest to $0$} of rank weight $w$.
  \Oubliettes{Its weight is the minimum distance $d$ of the code.}
  From {\bf Theorem \ref{thm:degann}}, there is a non-zero
  $\theta$-polynomial $U(X)$ of degree $w$ which vanishes on all the
  $c_i$'s.  There also exists a non-zero $\theta$-polynomial $f(X)$ of
  degree $\leq k-1$ such that $c_i = \Eval{f}{g_i}$, for all
  $i=1,\ldots,n$. Therefore, we have
        \[ 
        \Eval{U(X) \cdot f(X)}{g_i}=\Eval{U}{c_i}=0 \ForallInt{1}{i}{n}.
        \]
        Thus $U(X)\cdot f(X)$ is a non-zero $\theta$-polynomial of
        degree $\leq w+k-1$ which vanishes on the $K$-vector space of
        dimension $n$ generated by the $g_i$'s.  So we have
        $n \leq w+k-1$, i.e. $w \geq n-k+1$. Therefore $n-k+1 \leq
        d$. From Singleton bound we have $d\leq n-k+1$: $d=n-k+1$, and
        the code is MRD.  \qed\end{proof} Given the usual scalar
      product, we recall that the dual, or orthogonal, of a code is the set of vectors orthogonal to
      all the codewords.

\begin{proposition}
A parity-check matrix of  $ \Gab_{\theta,k}(\vectgr{g})$ has the form   
\begin{equation}
\label{Eq:FormeDuale}
       \matrice{H} = \MatriceCoins{\AppliqueTheta{h_1}{0}}{\AppliqueTheta{h_n}{0}}{\AppliqueTheta{h_1}{k-1}}{\AppliqueTheta{h_n}{k-1}},
\end{equation}
for some $h_1,\dots,h_n\in L^n$ which are $K$-linearly independent and satisfy
\[
		\MatriceCoins{\theta^{-d+2}(g_1)}{\theta^{-d+2}(g_n)}{\theta^{n-d}(g_1)}{\theta^{n-d}(g_n)} 
		\VecteurCoins{h_1}{h_n}=\VecteurCoins{0}{0}.
\]
Moreover the vector $(h_1, \ldots, h_n) \in L^n$ is unique up to a multiplicative factor in $L$.
\end{proposition}

\begin{proof}
	The equation $\matrice{G} \cdot {}^t\matrice{H} = 0$, where $\matrice{G}$ denotes the generating matrix of the code (\ref{eq:MatGenGab}) is equivalent to
	\[
	\sum_{i=1}^{n} \theta^a(g_i) \theta^b(h_i) = 0 , \quad 0 \leq a \leq k-1 , \quad 0 \leq b \leq d-2.
	\]
	By applying, for all  $b=0,\ldots,k-1$, the automorphism $\theta^{-b}$ we obtain
	\[
	\sum_{i=1}^{n} \theta^c(g_i) h_i = 0 , \quad -n+k+1 \leq c \leq  k-1,
	\]
	which gives the system
	\begin{equation}
          \label{eq:supportdual}
	\MatriceCoins{\theta^0(\tilde{g}_1)}{\theta^0(\tilde{g}_n)}{\theta^{n-2}(\tilde{g}_1)}{\theta^{n-2}(\tilde{g}_n)} 
	\VecteurCoins{h_1}{h_n}=\VecteurCoins{0}{0},
	\end{equation}
	with $\tilde{g}_i=\theta^{-(d-2)}(g_i)$.
	Since the $g_i$'s are $K$-linearly independent, so are the $\tilde{g}_i$'s. 
	Therefore, $\wthK(\vectgr{ \tilde g})$, the $K$-rank weight of $\vectgr{\tilde{g}}=(\tilde{g}_1,\ldots,\tilde{g}_n)$, is equal to $n$. 
	From {\bf Theorem \ref{prop:egaliterang}} the $L$-rank of the matrix 
	\[
	\MatriceCoins{\theta^0(\tilde{g}_1)}{\theta^0(\tilde{g}_n)}{\theta^{m-1}(\tilde{g}_1)}{\theta^{m-1}(\tilde{g}_n)}
	\]
        is also $n$.  Thus, the kernel of equation
        (\ref{eq:supportdual}) has  $L$-dimension
        $1$.  Let $\VecteurLigneCoins{h_1}{h_n}\neq0$ belong to the kernel.
        Now, we prove that the $h_i$'s are $K$-linearly independent. Consider a linear dependency
	\[ 
          \sum_{i=1}^n \lambda_i h_i = 0, \quad \lambda_i \in K.
          \]
	If the matrix, that we call  $A$, consisting in (\ref{eq:supportdual}) augmented with the row $(\lambda_1,\dots,\lambda_n)$,  as a last row, has rank $n$, 
	then the $h_i$'s are all zero. Since this is not the case, 
        $A$ has rank at most $n-1$ and there exists $\mu_0,\ldots, \mu_{n-2} \in L^{n-1}$ such that  
	\[
          \lambda_i = \mu_0 \theta^0(\tilde{g}_i) + \cdots + \mu_{n-2} \theta^{n-2}(\tilde{g}_i), 	\quad i=1,\ldots,n,
	\]
        since matrix in (\ref{eq:supportdual})  has rank $n-1$.
	Let  $M =\sum_{i=0}^{n-2} \mu_i \Monome{X}{i} \in L[X;\theta]$. 
	Since for all $i$, $\lambda_i \in K$, the $\lambda_i$'s are also roots of  $\Monome{X}{1}-\Monome{X}{0}$. 
        Therefore, from the properties of the evaluation of $\theta$-polynomials in {\bf Proposition \ref{prop:OperatorEvaluation}}, we have, for  $i=1,\ldots,n$:
	\[
        \Eval{(\Monome{X}{1}-\Monome{X}{0}) \cdot M}{\tilde{g}_i} = \Eval{(\Monome{X}{1}-\Monome{X}{0})}{\Eval{M}{\tilde{g}_i}} =  \Eval{(\Monome{X}{1}-\Monome{X}{0})}{\lambda_i} =0. 
        \]
        The $\theta$-polynomial $(\Monome{X}{1}-\Monome{X}{0}) \cdot M$ has degree $n-1$ with $n$ $K$-linearly independent roots. 
        Therefore, by {\bf Proposition \ref{crl:hdim}} it is the zero polynomial. This implies that all the $\mu_i$ are zero, and for  $i=1,\ldots,n,~\lambda_i = 0$.
	   Therefore the $h_i$'s are $K$-linearly independent.
\qed\end{proof}
\begin{corollary}\label{crl:dualMRD}
	The dual of a generalized Gabidulin code is a generalized Gabidulin code.
\end{corollary}


\section{Decoding Generalized Gabidulin Codes}

Let $K \hookrightarrow L$ be a cyclic Galois extension of finite degree
$m=[L:K]$ and $\theta$ a $K$-automorphism such that
$\Auto_K(L)=\langle \theta \rangle$.  We consider a generalized
Gabidulin code $\Gab_{\theta,k}(\vectgr{g})$ with parameters
$[n,k,d]_r$, for some support $\vectgr g=(g_1,\dots,g_n)\in L^n$, such
that the $g_i$'s are $K$-linearly independent.  The problem of
decoding the code $\Gab_{\theta,k}(\vectgr{g})$ is classically stated
as follows. Let $\vectgr{y}\in L^n$ be  such that
\[
\vectgr{y} = \vectgr{c}(f)  +  \vectgr{e} + \vectgr{\varepsilon}=\vectgr{c} +\vectgr{e} + \vectgr{\varepsilon},
\]
where $\vectgr c=\vectgr{c}(f)=\Eval{f}{\vectgr{g}} \in \Gab_{\theta,k}(\vectgr{g})$ is the codeword 
corresponding to the evaluation of $f(X)\in L[X;\theta]$, $\vectgr{e}$ and
$\vectgr{\varepsilon}$ are respectively the vector of errors and erasures, which are
``small''. The goal is to recover
$\vectgr c\in \Gab_{\theta,k}(\vectgr{g})$, or, equivalently, $f(X)\in L[X;\theta]$. 

For an error $\vectgr e$, being small means that $\vectgr e$ has low
rank weight. For the erasure $\vectgr\varepsilon$, there are two
definitions, which are more involved.  In the first part of this
Section, we introduce the decoding problem in presence of errors only
and show how to treat it by solving the reconstruction problem of
$\theta$-polynomials. Then we recall two models of
erasures introduced in~\cite{roth1991maximum,Blaum/McEliece:1985,li2014transform,koetter2008coding}
adapted to existing applications (line erasures and network coding
erasures for instance) and show how the problem of decoding in
presence of these errors and erasures can be reduced to the decoding problem of errors
in a  generalized Gabidulin code derived from the original one.

By choosing a fixed $K$-basis $\mathcal B$ of $K \hookrightarrow L$, we will
indifferently consider the vectors of $L^n$ as $m \times n$ matrices
over $K$.  

\subsection{Decoding errors}
\label{sssec:decrec}


We define formally the decoding problem as
\begin{definition}[Decoding problem $\Dec(n,k,t,\vectgr{g},\vectgr{y})$]\label{def:DP}\hfill
    \begin{itemize}
    \item  \textbf{Input:} 
      \begin{itemize}
        \item $n,k\in \N$, $k\leq n$;
       \item support $\vectgr g=(g_1,\dots,g_n)\in L^n$, where the $g_i$'s are
        $K$-linearly independent;
        \item  the  Gabidulin code $\Gab_{\theta,k}(\vectgr{g})$;
        \item  $t\in \N$, the rank weight of the error vector (or matrix);
        \item  $\vectgr{y}=\VecteurLigneCoins{y_1}{y_n} \in L^n$.
      \end{itemize}
      \item \textbf{Output:} 
      \begin{itemize}
        \item $(\vectgr c,\vectgr e)$ or, equivalently, $(f(X),\vectgr e)$, such that:
          \begin{enumerate}
      \item $\vectgr c\in \Gab_{\theta,k}(\vectgr{g})$, or, equivalently,
        $f(X)\in L[X;\theta]$, $\deg f(X)<k$;
        \item $\vectgr{e}=\VecteurLigneCoins{e_1}{e_n} \in L^n$ with $w(\vectgr{e}) \leq t$;
        \item $
        \vectgr y=\vectgr c+\vectgr e=\Eval{f}{\vectgr{g}}+\vectgr e
$,
      \end{enumerate}
    \item \textbf{or} return \textbf{fail} is no such solution exists.
      \end{itemize}
      \end{itemize}
\end{definition}

\Oubliettes{                  
This problem is directly related to our problem of: given a received vector $\vectgr{y}=\vectgr{c}(f)+\vectgr{e}$, 
where the rank of $\vectgr{e}$ is upper-bounded by t$t$ and $\vectgr{c}(f)$ is 
a codeword of $\Gab_{\theta,k}(\vectgr{g})$, we want to find $f$.  
It is clear that if $t$ is less than the error-correcting capability $\lfloor (n-k)/2 \rfloor$, then the solution is unique.
}

\begin{example}
In the context of Example~\ref{ex:filrouge}, suppose we  receive
\[
\begin{pmatrix}											
1	&	-1	&	1	&	3	&	-1	&	1	\\
1	&	0	&	-1	&	0	&	-1	&	-2	\\
0	&	1	&	0	&	-1	&	-1	&	-1	\\
0	&	1	&	1	&	1	&	0	&	0	\\
1	&	-1	&	3	&	2	&	-1	&	0	\\
0	&	1	&	-1	&	-1	&	-1	&	-2	
\end{pmatrix}
=
\begin{pmatrix}
0 & 0 & 0 & 1 & -1 &  0 \\ 
0 & 1 & 0 & 0 & -1 & -1 \\ 
1 & 0 & 0 & 0 & -1 & -1 \\ 
0 & 1 & 0 & 0 &  0 & -1 \\ 
0 & 0 & 2 & 0 & -1 & -1 \\ 
1 & 0 & 0 & 1 & -1 & -1
\end{pmatrix}
+
\begin{pmatrix}											
1	&	-1	&	1	&	2	&	0	&	1	\\
1	&	-1	&	-1	&	0	&	0	&	-1	\\
-1	&	1	&	0	&	-1	&	0	&	0	\\
0	&	0	&	1	&	1	&	0	&	1	\\
1	&	-1	&	1	&	2	&	0	&	1	\\
-1	&	1	&	-1	&	-2	&	0	&	-1	
\end{pmatrix}
\]
which can be written in vector form
\[ \VecteurLigneCoins{y_1}{y_6}=\VecteurLigneCoins{c_1}{c_6}+(\epsilon_1,-\epsilon_1,\epsilon_2,\epsilon_1+\epsilon_2,0,\epsilon_2), \]
where $\epsilon_1=1+\alpha-\alpha^2+\alpha^4-\alpha^5$
and $\epsilon_2=1-\alpha+\alpha^3+\alpha^4-\alpha^5$.
In this case we have an error of rank 2.
\end{example}
Our decoding method is inspired by the so-called Welch-Berlekamp algorithm~\cite{berlekamp1986error,gemmell1992highly},
and we first present  an intermediate  problem.
\begin{definition}[Non Linear Reconstruction problem
  $\NLR(n,k,t,\vectgr{g},\vectgr{y})$]\label{def:NLRP}\hfill

 \begin{itemize}
    \item  \textbf{Input:}
      \begin{itemize}
        \item $n,k\in\N$, $k\leq n$;
        \item   $t\in\N$, the number of errors;
        \item $\vectgr{g}=\VecteurLigneCoins{g_1}{g_n} \in L^n$, where
          the $g_i$'s are $K$-linearly independent;
        \item  $\vectgr{y}=\VecteurLigneCoins{y_1}{y_n} \in L^n$.
      \end{itemize}
      \item \textbf{Output:}
        \begin{itemize}
          \item $f(X),V(X)\in L[X;\theta]$ such that:
      \begin{enumerate} 
        \item  $\deg f(X) < k$;
        \item  $V(X)\neq 0$ and $ \deg V(X) \leq t$;
        \item
          $
             \Eval{V}{y_i} = \Eval{V \cdot f}{g_i}=\Eval{V}{\Eval{f}{g_i}}$, $ i=1,\dots,n
$,
      \end{enumerate}
    \item \textbf{or} return \textbf{fail} if no such solution exists.

   \end{itemize}
   \end{itemize}
\end{definition}
In the above definition, $V(X)$ plays the role of the locator polynomial
of the error: if $\vectgr y=\vectgr g+\vectgr e$, where $w(\vectgr{e}) \leq t$, then  $V(X)$
vanishes on the $K$-vector space generated by $\vectgr e$.

\begin{proposition}
  There is a one-to-one correspondence between
  solutions of $\Dec(n,k,t,\vectgr{g},\vectgr{y})$ and those of
  $\NLR(n,k,t,\vectgr{g},\vectgr{y})$.
\end{proposition}

\begin{proof}
  Let $n,k,t,x\vectgr{g},\vectgr{y}$ be the parameters of the two
  statements. If $f(X)\in L[X;\theta]$ and
  $\vectgr{e}\in L^n$ are solution of
  $\Dec(n,k,\vectgr{g},\vectgr{y})$, then $f(X)$ and
  $V(X) = \Ann_{\langle e_1, \ldots, e_n \rangle}(X)$ are a solution of
  $\NLR(n,k,t,\vectgr{g},\vectgr{y})$.
  Conversely, if $f(X)$ and $V(X)$ are a solution of
  $\NLR(n,k,\lfloor \frac{n-k}{2} \rfloor,\vectgr{g},\vectgr{y})$,
  then $f(X)$ and
  $\vectgr{e}=\VecteurLigneCoins{y_1-\Eval{f}{g_1}}{y_n-\Eval{f}{g_n}}$
  are a solution of $\Dec(n,k,\vectgr{g},\vectgr{y})$.  \qed
\end{proof}
Therefore, to decode errors it is sufficient to solve the Non Linear
Reconstruction problem. However since the related equations involve
products of unknowns, we would have to solve quadratic systems over $L$,
which is a difficult and intractable task.  In the case we are
interested in errors of weight less than the error-correcting
capability, we can ``linearize'' this problem, i.e. treat terms of
degree as new indeterminates, and then the solutions to the linearized
problem give the solution to the Non Linear Reconstruction problem
problem, thus enabling to decode.

\begin{definition}[Linear Reconstruction problem
  $\LR(n,k,t,\vectgr{g},\vectgr{y})$]\label{def:LRP}\hfill
\begin{itemize}
    \item  \textbf{Input:}
      \begin{itemize}
        \item $n\in\N$, $k\leq n$;
        \item   $t\in\N$ the number of errors;
        \item  $\vectgr{g}=\VecteurLigneCoins{g_1}{g_n} \in L^n$,  where
          the $g_i$'s are $K$-linearly independent;
        \item  $\vectgr{y}=\VecteurLigneCoins{y_1}{y_n} \in L^n$;
      \end{itemize}
      \item \textbf{Output:}
        \begin{itemize}
          \item $N(X), W(X)\in L[X;\theta]$ such that
      \begin{enumerate} 
      \item  $W(X)\neq 0$ 
      \item   $ \deg W(X) \leq t$ and 
        \begin{itemize} 
         \item either $\deg N(X) \leq k+ t - 1$ if $n-k$ is even;
         \item or $\deg N(X) \leq k+ t$ if $n-k$    is  odd;
          \end{itemize} 
       
\item $            \Eval{W}{y_i} = \Eval{N}{g_i}$, $i=1,\dots,n$.
      \end{enumerate}
    \item \textbf{or} return \textbf{fail} if no such solution exists.
      \end{itemize}
   \end{itemize}
\end{definition}
A solution $(V(X),f(X))$ of
$\NLR(n,k,t,\vectgr{g},\vectgr{y})$ gives a solution $(V(X) \cdot f(X),V(X))$
of $\LR(n,k,t,\vectgr{g},\vectgr{y})$.  We have the following converse
proposition:

\begin{theorem}\label{theorem:LRtoNLR}
  Let $t \leq \lfloor \frac{n-k}{2} \rfloor$. If there exists a
  solution $(V(X),f(X))$ to $\NLR(n,k,t,\vectgr{g},\vectgr{y})$, then
any solution $(W(X),N(X))$ of  $\LR(n, k, t, \vectgr{g}, \vectgr{y})$
  satisfies 
  \[
  W(X) \cdot f(X) = N(X).
  \]
\end{theorem}
	
\begin{proof} 
  Let $(V(X),f(X))$ be a solution of
  $\NLR(n,k,t,\vectgr{g},\vectgr{y})$, with $\deg V(X) \leq t$,
  $\deg f(X)\leq k-1$ and $V\neq 0$.  We set
  \[
    e_i= y_i - \Eval{f}{g_i},\quad i=1, \ldots, n,
  \]
  and $\vectgr e = (e_1,\dots,e_n)$. For
  all $i=1,\ldots,n$,      since $V(X)$ is non-zero and 
     $\vectgr e$ has rank  less than or equal to $ t $:
	\[ 
        \Eval{V}{e_i} = \Eval{V}{y_i-\Eval{f}{g_i}}=\Eval{V}{y_i}-\Eval{V\cdot f}{g_i} = 0.
     \]
Let
     $(N(X),W(X))$ be a solution of $\LR(n,k,t,\vectgr{g},\vectgr{y})$, with
 $W(X)\neq 0$ and  $\deg W(X) \leq t$. For all 
     $i=1,\ldots,n$,
	\[
        \begin{array}{rl}
        \Eval{W}{e_i}   &=\Eval{W}{y_i} - \Eval{W}{\Eval{f}{g_i}}      \\
	                &= \Eval{N}{g_i} - \Eval{W}{\Eval{f}{g_i}}
        \end{array}
         \]
         Since the vector $(\Eval{W}{e_1},\ldots\Eval{W}{e_n})$ has rank less than or equal to $ t$, there exists $U(X) \ne 0$ with degree
         less than or equal to $ t$ such that, for  $i=1,\ldots,n$,
	\[  
        \Eval{U}{\Eval{W}{e_i}} = \Eval{U}{\Eval{N}{g_i} - \Eval{W}{\Eval{f}{g_i}}} = 0.
        \]                
	Hence, for $i =1,\dots n$,
        \[\Eval{\left(U(X) \cdot (N(X) - W(X) \cdot f(X))\right)}{g_i} = 0.
        \]
        Now since $\deg N(X) \leq k+t-1$ if $n-k$ is even and  $\deg N(X) \leq k+t$ if $n-k$ is odd, 
  $U(X) \cdot (N(X) - W(X) \cdot f(X))$ is a
        $\theta$-polynomial of degree  less than or equal to $ k + 2t-1$ if $n-k$ is even and $k+2t$ if $n-k$ is odd. 
Since  $t \leq \lfloor \frac{n-k}{2} \rfloor$, in both cases, its degree at most  $ n-1$. 
Moreover it vanishes on $n$ $K$-linearly independent elements. Therefore
\[        
    U(X) \cdot (N(X) - W(X) \cdot f(X) )= 0
\]
Since
        $L[X;\theta]$ has no zero divisor, we have
        $N(X) = W(X) \cdot f(X)$.  
        \qed\end{proof} 
As an immediate consequence, provided $t \leq\lfloor (n-k)/2 \rfloor$,  
any solution $(W(X),N(X))$ of $\LR(n,k,t,\vectgr{g},\vectgr{y})$ gives the
      $\theta$-polynomial $f(X)$ solution to
      $\Dec(n,k,t,\vectgr{g},\vectgr{y})$, obtained  by a simple left Euclidean
      division $f(X) = W (X)\backslash N(X)$ in the $\theta$-polynomial ring
      $L[X;\theta]$.  This procedure is presented in {\bf Algorithm~\ref{algo:lr}}.

\begin{remark} 
  The difference we made between the oddity of $n-k$ is not essential
  here, but it is convenient to introduce this distinction that will
  be needed  in further algorithms.
\end{remark}

\begin{theorem}
Let
  $\vectgr{g}=\VecteurLigneCoins{g_1}{g_n} \in L^n$ be $K$-linearly
  independent elements.  Let
  $ \vectgr{y} = \vectgr{c}(f) + \vectgr{e}$, where
  $\vectgr{c}(f) \in \Gab_{\theta,k}(\vectgr{g})$.  If
  $w( \vectgr{e}) \leq  \lfloor \frac{n-k}{2} \rfloor$, then
  {\bf Algorithm~\ref{algo:dec}} recovers $(f,\vectgr{e})$  on inputs $(n,k,\vectgr{g},\vectgr{y})$.
        If $w(\vectgr{e})> \lfloor \frac{n-k}{2} \rfloor$, then   {\bf Algorithm~\ref{algo:dec}}  returns \textbf{\emph{fail}}.
\end{theorem}
\begin{proof}
  Since $w( \vectgr{e}) \leq \lfloor \frac{n-k}{2} \rfloor$, then
  there is a solution to
  $\NLR(n,k, \lfloor \frac{n-k}{2} \rfloor,\vectgr g,\vectgr y)$, and
  {\bf Theorem~\ref{theorem:LRtoNLR}} ensures that for
  $t \leq \lfloor \frac{n-k}{2} \rfloor$ any solution $(N(X),W(X))$ to
  $\LR(n,k,t,\vectgr g,\vectgr y)$ given by {\bf
    Algorithm~\ref{algo:lr}}, gives a solution to
  $\NLR(n,k,t,\vectgr g,\vectgr y)$, by a left Euclidean division.
  
  To prove the converse, suppose that {\bf Algorithm~\ref{algo:dec}}
  returns $(f(X),\vectgr{e})$. This implies that {\bf
    Algorithm~\ref{algo:lr}} has returned $N(X),W(X)$ with
  $\deg W(X) \leq\lfloor \frac{n-k}{2} \rfloor$; $N(X) = W(X) \cdot f(X)$, where
  $\deg f(X) \leq k$. This last equation gives
  $\Eval{W}{\Eval{f}{g_i}}=\Eval{N}{g_i}=\Eval{W}{y_i}$, {\em i.e.}
  $\Eval{W}{\Eval{f}{g_i}-y_i}=0$ $\ForallInt{1}{i}{n}$, which implies
  $w(\vectgr{e})\leq \lfloor \frac{n-k}{2} \rfloor$, since
  $\deg W(X)\leq \lfloor \frac{n-k}{2} \rfloor$. Since $\deg f(X)<k$,
  we have
  $\vectgr{c}=\Eval{f}{\vectgr{g}} \in \Gab_{\theta,k}(\vectgr{g})$.
  Thus there is a solution to the decoding problem.  Therefore, the
  algorithm returns a failure if and only if the rank of the error
  vector is larger than $ \lfloor \frac{n-k}{2} \rfloor$.
  \qed\end{proof}
\begin{algorithm}
	\caption{LinearReconstruct (via Gaussian elimination)}
	\label{algo:lr}
	\begin{algorithmic}[1]
		\ION \REQUIRE $n,k\in \N$, $k\leq n$;
		\IOB \REQUIRE $\vectgr{g} = (g_1, \ldots, g_n) \in L^n$, where the $g_i$'s are
                $K$-linearly independent;
                \IOB  \REQUIRE $\vectgr{y}=\VecteurLigneCoins{y_1}{y_n} \in L^n$.
		\ION \ENSURE $N(X),W(X) \in L[X;\theta]$ solution to $\LR(K,L,n,k,t= \lfloor \frac{n-k}{2} \rfloor,\vectgr g,\vectgr y)$
                \STATE Let $s = \lfloor (n-k)/2 \rfloor$ if $n-k$ is even and $s= \lfloor (n-k)/2 \rfloor +1$ otherwise
		\STATE Compute the kernel $V$ of the linear system~(\ref{eq:decGauss}), by Gaussian elimination
		\IF{ $V\neq \set0$}
                \STATE get a random non zero vector $	(n_0,\dots,	n_{k+s-1},	w_0,	\dots,	w_{t})\in V$
                \STATE collect the coefficient $n_i$'s in $N(X)$
                \STATE collect the coefficient $w_i$'s in $W(X)$
			\STATE return $(N(X),W(X))$
		\ELSE 	
			\STATE return \textbf{fail}
		\ENDIF
	\end{algorithmic}
\end{algorithm}
\begin{remark}\label{rmq:decGauss}
  Let $t = \lfloor (n-k)/2 \rfloor$ and $s = t-1$ if $n-k$ is even  and $s  = t$ otherwise. 
  Let $W(X)=\sum_{i=0}^{t} w_i X^i$, and $N(X)=\sum_{i=0}^{k+s-1} n_i X^i$, then the
  coefficients $w_i$'s and $n_i$'s of $W$ and $N$, solution of
  $\LR(n,k,t,\vectgr{g},\vectgr{y})$, satisfy
	\begin{equation}\label{eq:decGauss}
	\begin{pmatrix}
	g_1 & \cdots & \theta^{k+s-1}(g_1) & y_1 & \cdots & \theta^{t}(y_1) \\ 
	\vdots & \ddots & \vdots & \vdots & \ddots & \vdots \\ 
	g_n & \cdots & \theta^{k+s-1}(g_n) & y_n & \cdots & \theta^{t}(y_n)
	\end{pmatrix} 
	\cdot
	\begin{pmatrix}
	n_0 \\ 
	\vdots \\ 
	n_{k+s-1} \\ 
	-w_0 \\ 
	\vdots \\ 
	-w_{t}
	\end{pmatrix} 
	=
	0.
	\end{equation}
        It is a system of $k+s+t+1$ unknowns in $n$.
        The right kernel of the system  can be found in
        $O(n^3)$ arithmetic operations in $L$, by standard Gaussian
        elimination.
\end{remark}

\begin{algorithm}
	\caption{Decoding Gabidulin Codes up to half the minimum distance}
	\label{algo:dec}
	\begin{algorithmic}[1]        
		\ION \REQUIRE $n,k\in \N$, $k\leq n$;
		\IOB \REQUIRE $\vectgr{g} = (g_1, \ldots, g_n) \in L^n$, where the $g_i$'s are
                $K$-linearly independent;
                \IOB  \REQUIRE $\vectgr{y}=\VecteurLigneCoins{y_1}{y_n} \in L^n$.
		\ION \ENSURE $f(X) \in L[X;\theta]$, $\vectgr e\in L^n$
		\STATE $(N(X),W(X))\gets \text{LinearReconstruct}(K,L,\theta,n,k,\vectgr g,\vectgr y)$
                \IF{$\text{LinearReconstruct}(n,k,\vectgr g,\vectgr y)$ returns \textbf{fail}}
                \STATE return \textbf{fail}
                \ELSE
		\STATE $(f(X),R(X)) \longleftarrow \text{LeftEuclideanDivision}(N(X),W(X))$,\\
 i.e. $N(X)=W(X)f(X)+R(X)$
		\IF{$R(X)=0$ and $\deg f(X)<k$}
                \STATE return $f(X),\vectgr y-\Eval{f}{\vectgr{g}}$ 
		\ELSE 	
			\STATE return \textbf{fail}
		\ENDIF
		\ENDIF
	\end{algorithmic}
\end{algorithm}


	\subsection{Decoding errors and erasures}
\label{sssec:eff}

The notion of erasures is somewhat difficult to introduce. A first
model of ``line erasures''~\cite{roth1991maximum,Blaum/McEliece:1985}
considers the received  word as a matrix (all vectors in $L^n$ are expanded into matrices of size
$m \times n$ over $K$ by expanding each coefficient over a $K$-basis
$\mathcal{B}= ( b_1,\ldots,b_m )$ of $L$.): 
\[
  \matrice{Y} = \matrice{C} + \matrice{E} + \matrice{\mathcal{E}} \in (K\cup{?})^{m \times n},
\]
where $\matrice{C} \in \Gab_{\theta,k}(\vectgr{g})$, $\matrice{E}$ is
the error-matrix of rank weight $\leq t$ and
$\matrice{\mathcal{E}}$ is the ``erasures'' matrix, with ``$?$'' being the
erasure symbol.
Matrices can not be expressed as vectors in this model.

A second model considers
network coding erasures~\cite{li2014transform,koetter2008coding}. It
is a bit more involved and does not fit with the above short
introduction, and shall be introduced later.
\subsubsection{Decoding in presence of line erasures}
Some coefficients of $\matrice{Y}$ are erased. 
We model these erasures by a matrix $\vectgr{\mathcal{E}}$ whose coefficients are $?$ or $0$.
The $?$ correspond to erased coefficients, with the convention that for any $x\in L$, $?+x=?$.
Since the rank of a such matrix is not well-defined,  its weight is measured by the \emph{term-rank}.
\begin{definition}
  The \emph{term-rank} $w_{tr}(\mathcal{E})$ of a matrix $\mathcal{E}$ is the minimal
  size of a set $S$ of rows or columns such that any non-zero
  entry of the matrix $\mathcal{E}$ belong to a row or a column of $S$. We
  denote by $S_r$ and $S_c$ the indices of rows and columns of a
  minimal covering of the erasure matrix, i.e.
  \[
    \vectgr{\mathcal{E}}_{i,j} = ? \implies i \in S_r \text{ or } j
  \in S_c,
\]
furthermore, let $s_r=\size{S_r}$ and $s_c=\size{S_c}$.
\end{definition}
Notice that coverings, even minimal, are not unique.
This was originally the metric considered for applications of Gabidulin codes~\cite{roth1991maximum}. 
\begin{example}
We receive 
\begin{align}
\begin{pmatrix}											
1	&	0	&	?	&	2	&	0	&	0	\\
1	&	-1	&	0	&	1	&	0	&	-1	\\
0	&	1	&	?	&	-1	&	-2	&	-1	\\
?	&	1	&	?	&	?	&	0	&	-1	\\
1	&	-1	&	2	&	1	&	0	&	-1	\\
0	&	1	&	?	&	0	&	-2	&	-1	
\end{pmatrix}
=
	\begin{pmatrix}
		0 & 0 & 0 & 1 & -1 &  0 \\ 
		0 & 1 & 0 & 0 & -1 & -1 \\ 
		1 & 0 & 0 & 0 & -1 & -1 \\ 
		0 & 1 & 0 & 0 &  0 & -1 \\ 
		0 & 0 & 2 & 0 & -1 & -1 \\ 
		1 & 0 & 0 & 1 & -1 & -1
	\end{pmatrix}
+
\begin{pmatrix}											
1	&	-1	&	0	&	1	&	1	&	0	\\
1	&	-1	&	0	&	1	&	1	&	0	\\
-1	&	1	&	0	&	-1	&	-1	&	0	\\
0	&	0	&	0	&	0	&	0	&	0	\\
1	&	-1	&	0	&	1	&	1	&	0	\\
-1	&	1	&	0	&	-1	&	-1	&	0	
\end{pmatrix}
+
\begin{pmatrix}
0	&	0	&	?	&	0	&	0	&	0	\\
0	&	0	&	0	&	0	&	0	&	0	\\
0	&	0	&	?	&	0	&	0	&	0	\\
?	&	0	&	?	&	?	&	0	&	0	\\
0	&	0	&	0	&	0	&	0	&	0	\\
0	&	0	&	?	&	0	&	0	&	0
\end{pmatrix}.
\end{align}
In this example, the erasure matrix has term-rank $2$, since the third column and the fourth row enable to cover all $?$ symbols.
\end{example}
\begin{definition}[Decoding problem with line erasures]\hfill

    \begin{itemize}
    \item  Input: 
      \begin{itemize}
        \item  $\Gab_{\theta,k}(\vectgr{g})$, with parameters $[n,k,d]_r$; 
        \item  $\vectgr{Y} \in \mathcal{M}_{m \times n}(K)$;
        \item  $\vectgr{\mathcal{E}}$ the erasure matrix, whose term-rank is denoted by $s$.
      \end{itemize}
    \item Output:
      \begin{itemize}
        \item $f \in L[X;\theta]$ with $\deg f(X) < k$;
        \item $\vectgr{E} \in \mathcal{M}_{m \times n}(K)$ with $w(\vectgr{E}) \leq \lfloor \frac{n-k}{2} \rfloor -s$,
        such that for all $i,j$, such that $\vectgr{\mathcal{E}}_{i,j} \neq ?$, 
        \[
          \vectgr{Y}_{i,j} = \vectgr{C}(f)_{i,j} + \vectgr{E}_{i,j}.
        \]
      \end{itemize}
    \end{itemize}
\end{definition}
The procedure for correcting errors and erasures consists of
eliminating erasures by reducing the problem of decoding error and
erasures in the initial Gabidulin code to decoding errors of a derived
Gabidulin code.  We show in the following steps how to deal with
column erasures first, then with line erasures, and finally decode
using the classical model of rank errors.
\begin{enumerate}
\item Column erasures: Let
  $\vectgr{\widetilde{y}} = (y_i)_{i \notin S_c}$, be the received
  vector punctured on the erased columns. This is a vector of length
  $n-s_c$ satisfying
\[
 \vectgr{\widetilde{y}} = \vectgr{\widetilde{c}}(f(X)) +  \vectgr{\widetilde{e}} + \vectgr{\widetilde{\varepsilon}},
\]
where $\tilde e$ and $\tilde\varepsilon$ are the punctured version of
$e$ and $\varepsilon$. Since $\vectgr{\widetilde{\varepsilon}}$ has no
more column erasures, the problem is thus now reduced to decoding
errors and rows erasures in $\Gab_{\theta,k}(\vectgr{\widetilde{g}})$;
\item Row erasures: Let
  $\mathcal{V}_r(X) = \Annul{\langle b_i, i \in S_r \rangle}$ be the
  $\theta$-polynomial annihilating the elements of $\mathcal{B}$
  labeled by $S_r$. This implies in particular that
  $\Eval{\mathcal{V}_r}{\vectgr{\widetilde{\varepsilon}}} = 0$. Therefore,
  since
  $ \vectgr{\widetilde{c}}(f(X)) = \Eval{f}{\vectgr{\widetilde{g}}}$ for
  some $\theta$-polynomial $f(X)$ of degree less than or equal to $ k-1$, we obtain:
\[
   \Eval{\mathcal{V}_r}{\vectgr{\widetilde{y}}} =   \Eval{\mathcal{V}_r \cdot f}{\vectgr{\widetilde{g}}}   +  \Eval{\mathcal{V}_r}{\vectgr{\widetilde{e}}}.
\]   
The problem is now reduced to decoding errors in $\Gab_{\theta,k+s_r}(\vectgr{\widetilde{g}})$.
\item Correcting errors: Since the rank of
  $\Eval{\mathcal{V}_r}{\vectgr{\widetilde{e}}}$ is at most the rank
  of $\vectgr{\widetilde{e}}$, we can correct the remaining errors by
  solving
  $\LR(n-s_c,k+s_r,\lfloor\frac{(n-s_c)-(k+s_r)}{2}\rfloor,\vectgr{g},\vectgr{y})$
  as shown in previous Section. 
  The rank of the new error should by at most the error capability of the new code.  
  After this step, $f(X)$ is
  recovered by a Euclidean division on the left by $\mathcal{V}_r(X)$ in
  $L[X;\theta]$.
\end{enumerate}
The procedure is given in {\bf Algorithm \ref{algo:eras}}. 
\begin{theorem}
If  $2t+s_r+s_c \leq n-k$, then the $\theta$-polynomial $f$ can be uniquely recovered by {\bf Algorithm \ref{algo:eras}} on inputs
 $n, k, S_r, S_c, \vectgr{\widetilde{g}}, \vectgr{\widetilde{y}}, \mathcal{B}$.
\end{theorem}
\begin{proof}
	The third step succeeds if the number of original errors 
	is lower or equal to the error correcting capability 
	of the code that we get after steps 1 and 2, 
	which has length $n-s_c$ and dimension $k+s_r$.
	Thus, decoding succeeds if 
	\[
		w(\Eval{\mathcal{V}_r}{\vectgr{\widetilde{e}}})
		 = t' \leqslant 
		\lfloor \frac{(n-s_c)-(k+s_r)}{2} \rfloor
	\] 
	\textit{i.e.} if 
	\[ 2t' + s_c + s_r \leqslant n-k. \]
	Note that $t'=w(\Eval{\mathcal{V}_r}{\vectgr{\widetilde{e}}}) \leqslant w(\vectgr{e})=t$, 
	thus $2t + s_c + s_r \leqslant n-k$ is sufficient for ensuring successful decoding.
	In practice, the weight of the error does not decrease when correcting the erasures, hence $t=t'$.

\qed
\end{proof}

\begin{algorithm}
\caption{Decoding line erasures}
\label{algo:eras}
\begin{algorithmic}[1]
\ION \REQUIRE A $K$-basis $\mathcal B=(b_1,\ldots,b_m) \in L^m$ of $L$
\IOB \REQUIRE $\vectgr{g} = (g_1, \ldots, g_n) \in L^n$
\IOB \REQUIRE $k$ the dimension of the code
\IOB \REQUIRE $\vectgr{y} = (y_1, \ldots, y_n) \in L^n$ 
\IOB \REQUIRE $S_c$ and $S_r$, sets of distinct labels of size $s_c$ and $s_r$, such that $s_r+s_c\leq n-k$
\ION \ENSURE  $f \in L[X;\theta]$
\STATE $\vectgr{\widetilde{g}} \longleftarrow (g_i, i \not\in S_c)$
\STATE $\vectgr{\widetilde{y}} \longleftarrow (y_i, i \not\in S_c)$
\STATE $\mathcal{V}_r(X) \longleftarrow \Annul{\langle b_i, i \in S_r \rangle}$
\STATE $z_i \longleftarrow \Eval{\mathcal{V}_r}{y_i},~i \notin S_c$
\STATE $\vectgr{\widetilde{z}} \longleftarrow (z_i, i \not\in S_c)$
\STATE $F(X) \longleftarrow$ {\bf Algorithm \ref{algo:dec}}($n-s_c,k+s_r,\vectgr{\widetilde{g}},\vectgr{\widetilde{z}}$)
\RETURN $f(X) = \mathcal{V}_r(X) \backslash F(X)$
\end{algorithmic}
\end{algorithm}


	\subsubsection{Decoding network coding erasures}
\label{sssec:effsido} 

The erasure model proposed in \cite{li2014transform} is related to the problem of correcting errors and erasures
in network coding applications of Gabidulin codes.  The received vector is seen in matrix form, in the $K$-basis $\mathcal B$:
\[
\matrice{Y} = \matrice{C}(f) + \matrice{E} + \matrice{\widehat{A_r}} \cdot \matrice{B_r} + \matrice{A_c} \cdot \matrice{\widehat{B_c}},
\] 
where $\matrice{\widehat{A_r}}$ and $\matrice{\widehat{B_c}}$ are known to the receiver.
Denoting by $s_r$ and $s_c$ the number of row and column erasures and by $r$ the rank of the full error, right hand side matrices have sizes
$m \times n$, $m \times n$, $m \times s_r$, $s_r \times n$, $m \times s_c$ and $s_c \times n$.

\begin{example}
Here is an example of a codeword altered by a full error, $1$ row erasure and $1$ column erasure. The bold-written matrices are  known to the receiver.
\[										
\matrice Y=	\begin{pmatrix}
		0 & 0 & 0 & 1 & -1 &  0 \\ 
		0 & 1 & 0 & 0 & -1 & -1 \\ 
		1 & 0 & 0 & 0 & -1 & -1 \\ 
		0 & 1 & 0 & 0 &  0 & -1 \\ 
		0 & 0 & 2 & 0 & -1 & -1 \\ 
		1 & 0 & 0 & 1 & -1 & -1
	\end{pmatrix}											
+											
	\begin{pmatrix}											
		1&-1&0&1&1&0\\
		1&-1&0&1&1&0\\
		-1&1&0&-1&-1&0\\
		0&0&0&0&0&0\\
		1&-1&0&1&1&0\\
		-1&1&0&-1&-1&0
	\end{pmatrix}						
+											
	\vec{
		\begin{pmatrix}											
			1	\\	-1	\\	0	\\	1	\\	1	\\	-1	
		\end{pmatrix}
	}											
	\cdot 											
	\begin{pmatrix}											
		0	\\	0	\\	1	\\	1	\\	0	\\	1	
	\end{pmatrix}^{\!\!\!\!\!\!\!\!\!\!\!\!\!\!\!t\,\,\,\,\,\,\,\,\,\,\,\,\,\,\,}											
+											
	\begin{pmatrix}											
		1	\\	1	\\	-1	\\	1	\\	-1	\\	1	
	\end{pmatrix}											
	\cdot 	
	\vec{										
		\begin{pmatrix}											
			1	\\	0	\\	-1	\\	0	\\	0	\\	1	
		\end{pmatrix}^{\!\!\!\!\!\!\!\!\!\!\!\!\!\!\!\!\!\!\!\!t\,\,\,\,\,\,\,\,\,\,\,\,\,\,\,\,\,\,\,\,}
	}.											
\]
\end{example}
The associated decoding problem is then defined as follows.
\begin{definition}[Decoding problem with network coding erasures]\hfill

    \begin{itemize}
    \item  Input: 
      \begin{itemize}
        \item  $\Gab_{\theta,k}(\vectgr{g})$, with parameters $[n,k,d]$; 
        \item  $\vectgr{Y} \in \mathcal{M}_{m \times n}(K)$;
        \item  $\matrice{\widehat{A_r}} \in \mathcal{M}_{m \times s_r}(K)$;
        \item  $\matrice{\widehat{B_c}} \in \mathcal{M}_{s_c \times n}(K)$.
      \end{itemize}
    \item Output:
      \begin{itemize}
        \item $f \in L[X;\theta]$ with $\deg(f) < k$;
        \item $\vectgr{E} \in \mathcal{M}_{m \times n}(K)$ with $w(\vectgr{E}) \leq \lfloor \frac{n-k}{2} \rfloor -(s_r+s_c)$;
        \item $\vectgr{A_c} \in \mathcal{M}_{s_r \times n}(K)$;
        \item $\vectgr{B_r} \in \mathcal{M}_{m \times s_c}(K)$;\\
        such that 
        \[
        \matrice{Y} = \matrice{C}(f) + \matrice{E} + \matrice{\widehat{A_r}} \cdot \matrice{B_r} + \matrice{A_c} \cdot \matrice{\widehat{B_c}}.
        \]
      \end{itemize}
    \end{itemize}
\end{definition}
The procedure for correcting errors and erasures is the following: 
\begin{enumerate}
 \item Column erasures:  
   let $\matrice{U}$ be a $n \times n$ matrix of rank $n$ such that
   $\matrice{\widehat{B_c}}\cdot \matrice{U}$ has its $n-s_c$ last
   columns equal to $\vectgr{0}$.  This matrix $\matrice{U}$ exists
   and corresponds to the column operations that would be applied to reduce
   $\matrice{\widehat{B_c}}$ to a column echelon form by Gaussian
   elimination.  We have
\[
	\matrice{Y} \cdot \matrice{U} 
	= \matrice{C} \cdot \matrice{U} 
	+ \matrice{E} \cdot \matrice{U} 
	+  \matrice{\widehat{A_r}} \cdot \matrice{B_r} \cdot \matrice{U}
	+  \matrice{\widehat{A_c}} \cdot \matrice{B_c} \cdot \matrice{U}.
\]
Let $\vectgr{\widetilde{g}}$ be the $n-s_c$ last positions of vector
$\vectgr{g}\matrice{U}$, $\vectgr{\widetilde{y}}$ be the $n-s_c$ last
positions of vector $\vectgr{y}\matrice{U}$ and
$\vectgr{\widetilde{e}}$ the $n-s_c$ last positions of vector
$\vectgr{e}\matrice U$.  Since puncturing
$\matrice{\widehat{A_c}} \cdot \matrice{B_c} \cdot \matrice{U}$ gives
$0$, the problem is thus now reduced to decoding errors and rows
erasures in $\Gab_{\theta,k}(\vectgr{\widetilde{g}})$ with the new
``received word'' $\matrice{\widetilde y} $.

\item Row erasures: Let $\mathcal{V}_r(X) \in L[X;\theta]$ be the
  annihilator of the $K$-vector space generated by the $s_r$ columns
  of $\matrice{\widehat{A_r}}$. Then the problem rewritten under
  vector form becomes
\[
\Eval{\mathcal{V}_r}{\vectgr{\widetilde{y}}} = \Eval{\mathcal{V}_r \cdot f}{\vectgr{\widetilde{g}}} +  \Eval{\mathcal{V}_r}{\vectgr{\widetilde{e}}}.
\]
Since $\mathcal{V}_r(X)$ has degree $s_r$,
the problem is now reduced to decoding errors in the Gabidulin code $\Gab_{\theta,k+s_r}(\vectgr{\widetilde{g}})$; 

\item Correcting errors: Since the rank of
  $\Eval{\mathcal{V}_r}{\vectgr{\widetilde{e}}}$ is at most the rank
  of $\vectgr{\widetilde{e}}$, this can be done by solving
  $\LR(n-s_c,k+s_r,\lfloor\frac{(n-s_c)-(k+s_r)}{2}\rfloor,\vectgr{\widetilde{g}},\vectgr{\widetilde{y}})$
  as shown in previous section.
  The rank of the new error should by at most the error capability of the new code.  
  The output of the error decoding
  algorithm is actually $\mathcal{V}_r(X)\cdot f(X)$, and $f(X)$ is recovered by a division on
  the left by $\mathcal{V}_r(X)$.
\end{enumerate}
The procedure is  detailed in {\bf Algorithm \ref{algo:erassido}}. 
\begin{theorem}
  If $2t+s_r+s_c \leq n-k$, then {\bf Algorithm
    \ref{algo:erassido}} uniquely recovers the codeword or the message
  polynomial, on inputs
  $n, k, \vectgr{g}, \vectgr{y}, \matrice{\widehat{B_c}},
  \matrice{\widehat{A_r}}$, of size given in the network coding erasures model.
\end{theorem}

\begin{proof}
	The proof is the same that in the first erasure model.
\qed
\end{proof}

\begin{algorithm}
\caption{Decoding network coding erasures}
\label{algo:erassido}
\begin{algorithmic}[1]
\ION \REQUIRE $\vectgr{g} =(g_1, \ldots, g_n) \in L^n$
\IOB \REQUIRE $\vectgr{y} =(y_1, \ldots, y_n) \in L^n$
\IOB \REQUIRE $\matrice{\widehat{B_c}} \in K^{s_c \times n}$ and $\matrice{\widehat{A_r}} \in K^{m \times s_r}$,
\IOB \REQUIRE $k$ the dimension of the code
\IOB \REQUIRE $s_r$ (resp.\ $s_c$) the number of row (resp.\ column) erasures,  such that $s_r+s_c\leq n-k$
\ION \ENSURE $f(X) \in L[X;\theta]$
\STATE Find $\matrice{U}\neq0 \in M_n(K)$ such that $\matrice{\widehat{B_c}} \cdot U$ is zero on the $n-s_c$ last columns 
\STATE $\vectgr{\widetilde{g}} \longleftarrow (\vectgr{g}\matrice{U})_{[s_c+1..n]} $
\STATE $\vectgr{\widetilde{y}} \longleftarrow (\vectgr{y}\matrice{U})_{[s_c+1..n]} $
\STATE $\mathcal{V}_r(X) \longleftarrow$  annihilator of the column space of $\matrice{\widehat{A_r}}$
\STATE $\widetilde{z}_i \longleftarrow \Eval{\mathcal{V}_r}{\widetilde{y}_i}$
\STATE $\vectgr{\widetilde{z}} \longleftarrow (z_i)_{[s_c+1..n]}$
\STATE $F(X) \longleftarrow$ {\bf Algorithm \ref{algo:dec}}($n-s_c,k+s_r,\vectgr{\widetilde{g}},\vectgr{\widetilde{z}}$)
\RETURN $f(X)= \mathcal{V(X)}_r \backslash F(X)$
\end{algorithmic}
\end{algorithm}


\section{The \TA}
	\label{part:ta}
The decoding algorithm presented in the previous section is based upon the resolution of the so-called linear reconstruction problem presented
in {\bf Definition~\ref{def:LRP}}. In this section we present an algorithm derived from the Welch-Berlekamp one \cite{berlekamp1986error}, 
in the version formulated by Gemmel and Sudan \cite{gemmell1992highly}, solving this linear reconstruction problem.  
Originally used for the decoding of Reed-Solomon codes, this algorithm  
was adapted to Gabidulin codes over finite fields by Loidreau \cite{loidreau2006welch}. The version that we present in this section is
a generalization of the latter one. It also works on finite field extensions by replacing the automorphism $\theta$ by the Frobenius automorphism. 
Our version takes into account cases that were not covered by the algorithm in \cite{loidreau2006welch}.

In a first section we present the algorithm. In a second one we prove that it indeed solves the
linear reconstruction problem. Then we study the complexity of decoding rank errors up to the error-correcting capability, by using this algorithm together with the left Euclidean division of $\theta$-polynomials. We show that it is always quadratic. Finally we present two variants of the decoding and precise their effects on the complexity.  

\subsection{The \TA}

The idea is to compute two  pairs $(N_0,W_0)$ and $(N_1,W_1)$ of $\theta$-polynomials which 
satisfy the interpolation conditions of the problem $\LR(n,k,\lfloor (n-k)/2 \rfloor,\vectgr{g},\vectgr{y})$ (see {\bf Definition~\ref{def:LRP}}): 
   \begin{equation}\label{Eq:IntCond}
            \Eval{W}{y_i} = \Eval{N}{g_i} \ForallInt{1}{i}{n},  
   \end{equation}
and such that at least one of the pairs  satisfies  the final degree conditions: 

\begin{equation}\label{Eq:DegCond}
       \begin{array}{l}
           \deg(N) \leq  
           \begin{cases}
           k+\lfloor\frac{n-k}{2} \rfloor -1, &\mbox{~if~}n-k{~even}, \\
           k+\lfloor\frac{n-k}{2} \rfloor,   &\mbox{~if~}n-k{~odd}, 
           \end{cases}\\
           -\infty < \deg(W)  \leq t.    
       \end{array}
\end{equation}

Therefore our aim is to control the growth of the degrees of the respective polynomials, but ensuring that at 
each round in the loop, the interpolation conditions are satisfied. 
 {\bf Algorithm \ref{algo:ta}} will be analyzed in the next section but we present here the ideas: 
\begin{enumerate}
        \item {\em Initialization step:} from lines $2$ to $7$. There are two way of constructing pairs of polynomials of relatively small degree  satisfying the  interpolation condition at step $k$.   
          \begin{itemize}
          \item One of the pair is formed with $(N_0 = \Annul{<g_1, \ldots, g_k>}, W_0 = $0$)$, 
where $\Annul{<g_1, \ldots, g_k>}$ is the annihilator polynomial, see {\bf Definition~\ref{df:annil}}. In that case we have $\deg(N_0) = k$. Note that this pair does not satisfy the degree conditions since $W_0 = 0$ and $\deg(N_0) = k > k-1$.  
  
          \item The other one is formed with $(N_1 =  \Interpol{[g_1, \ldots, g_k]}{[y_1, \ldots, y_k]}, W_1 = X)$, where $\Interpol{[g_1, \ldots, g_k]}{[y_1, \ldots, y_k]}$, is the interpolating polynomial, see {\bf Definition~\ref{def:interp}}. In that case 
$\deg(N_1) = k-1$ and $\deg(W_1) = 0$, therefore, this pair of polynomials satisfy the degree conditions. 
\end{itemize}

   \item  {\em Interpolation step (rounds $k+1 \leq j \leq n$):} lines $8$ to $45$.  
From two pairs of polynomials satisfying the interpolation condition and such that at least one of them satisfies the degree conditions at round $j$, 
\[
 \Eval{W}{y_i} = \Eval{N}{g_i} \ForallInt{1}{i}{j},  
\]
we construct two pairs of polynomials satisfying 
\[
 \Eval{W}{y_{i}} = \Eval{N}{g_{i}} \ForallInt{1}{i}{j+1}
\]
such  that at least one of the pairs satisfies the degree conditions at round $j+1$. 
To ensure this, we evaluate the discrepancy vectors $\vectgr{u} = (u_i)_{i=1}^n$, computing the difference vectors 
$ (\Eval{N}{g_i} - \Eval{W}{y_i})_{i=1}^n$ at every round of the loop. At round $i$, it must satisfy 
\[
\vectgr{u} = (0,\ldots,0,u_{i+1},\ldots,u_n)
\]  
This controls the effectiveness of the interpolation condition. To make  sure that at least one of the pairs satisfies the degree conditions we  increase the degree of  one pair on average by one every two rounds thanks to the updates presented in 
{\bf Table \ref{tbl:updates}}.
In  part~\ref{part:prta} we prove that one of  the obtained pairs satisfies the degree requirements (\ref{Eq:DegCond}).
\end{enumerate} 
\begin{algorithm}
\caption{Reconstruction algorithm}
\label{algo:ta}
\begin{algorithmic}[1]
\ION \REQUIRE  $k,n \in \N,~k \leq n$
\IOB \REQUIRE $\vectgr{g} = (g_1, \ldots, g_n) \in L^n$, $K$-linearly independent elements
\IOB \REQUIRE $\vectgr{y} = (y_1, \ldots, y_n) \in L^n$
\ION \ENSURE $N$ and $W$ solutions to 
\STATE \# Initialization step
\STATE $N_0(X) \longleftarrow \Annul{<g_1, \ldots, g_k>}$
\STATE $W_0(X) \longleftarrow 0$
\STATE $N_1(X) \longleftarrow \Interpol{[g_1, \ldots, g_k]}{[y_1, \ldots, y_k]}$ 
\STATE $W_1(X) \longleftarrow \Monome{X}{0}$
\STATE $\vectgr{u}_0= (u_{0,i})_{i=1}^n \longleftarrow  \Eval{N_0}{\vectgr{g}} - \Eval{W_0}{\vectgr{y}}$
\STATE $\vectgr{u}_1= (u_{1,i})_{i=1}^n \longleftarrow  \Eval{N_1}{\vectgr{g}} - \Eval{W_1}{\vectgr{y}}$
\STATE
\STATE \# Interpolation step 
\FOR {$i$ from $k+1$ to $n$}
	
	\STATE \# Secondary loop
	\STATE $j \longleftarrow i$
	\WHILE {$u_{0,j} \neq 0$ and $u_{1,j} = 0$ and $j \leq n$}
		\STATE $j \longleftarrow j+1$
	\ENDWHILE
	\IF {j=n+1}
		\STATE return $(N_1,W_1)$ \label{line:early}
	\ELSE
        \STATE   \# Permutation of the components of the positions $i$ and $j$
        \STATE $i \leftrightarrow j$
	\ENDIF
%
%
	\STATE
	\STATE \# Updates of  $\theta$-polynomials, according to discrepancies
	\IF {$u_1 \neq 0$}
		\STATE $N'_1 \longleftarrow (\Monome{X}{1} - \frac{\theta(u_{1,i})}{u_{1,i}} \Monome{X}{00}) \cdot N_1$
		\STATE $W'_1 \longleftarrow (\Monome{X}{1} - \frac{\theta(u_{1,i})}{u_{1,i}} \Monome{X}{00}) \cdot W_1$
		\STATE $N'_0 \longleftarrow N_0 - \frac{u_{0,i}}{u_{1,i}} N_1$
		\STATE $W'_0 \longleftarrow W_0 - \frac{u_{0,i}}{u_{1,i}} W_1$
\Oubliettes{
		\STATE \textcolor{green}{$\vectgr{u}'_1 \longleftarrow \theta(\vectgr{u}_1)-\frac{\theta(u_{1,i})}{u_{1,i}} \vectgr{u}_1$}
		\STATE \textcolor{green}{$\vectgr{u}'_0 \longleftarrow \vectgr{u}_0-\frac{u_{0,i}}{u_{1,i}} \vectgr{u}_1$}
}
	\ENDIF	
	\IF {$u_0 = 0$ and $u_1 = 0$}
		\STATE $N'_1 \longleftarrow \Monome{X}{1} \cdot N_1$
		\STATE $W'_1 \longleftarrow \Monome{X}{1} \cdot W_1$
		\STATE $N'_0 \longleftarrow N_0 $
		\STATE $W'_0 \longleftarrow W_0 $
\Oubliettes{
		\STATE \textcolor{green}{$\vectgr{u}'_1 \longleftarrow \theta(\vectgr{u}_1)$}
		\STATE \textcolor{green}{$\vectgr{u}'_0 \longleftarrow \vectgr{u}_0$}
}
	\ENDIF
\STATE
\STATE $N_0 \longleftarrow N'_1$
\STATE $W_0 \longleftarrow W'_1$
\STATE $N_1 \longleftarrow N'_0$
\STATE $W_1 \longleftarrow W'_0$
\Oubliettes{
\STATE \textcolor{green}{$\vectgr{u}_0 \longleftarrow \vectgr{u}'_0$}
\STATE \textcolor{green}{$\vectgr{u}_1 \longleftarrow \vectgr{u}'_1$}
}
\STATE	
       \STATE \Oubliettes{\textcolor{green}{ON ENLEVE LA LIGNE SUIVANTE :}}\# Discrepancies updates
       \STATE \Oubliettes{\textcolor{green}{ON ENLEVE LA LIGNE SUIVANTE :}}$\vectgr{u}_0= (u_{0,i})_{i=1}^n \longleftarrow  \Eval{N_0}{\vectgr{g}} - \Eval{W_0}{\vectgr{y}}$
       \STATE \Oubliettes{\textcolor{green}{ON ENLEVE LA LIGNE SUIVANTE :}}$\vectgr{u}_1= (u_{1,i})_{i=1}^n \longleftarrow  \Eval{N_1}{\vectgr{g}} - \Eval{W_1}{\vectgr{y}}$
\ENDFOR	

\STATE
\STATE return $N_1, W_1$

\end{algorithmic}
\end{algorithm}


\begin{table}[!h]
\[
\begin{array}{|c|c|r|l|c|}
\hline 
 u_{0,i}    & u_{1,i}    & A'_1                                                     		 & A'_0							& \text{type}\\ 
\hline 
* & \neq 0 & \left(\Monome{X}{1} - \frac{\theta(u_{1,i})}{u_{1,i}} \Monome{X}{00} \right) \cdot A_1 & A_0 - \frac{u_{0,i}}{u_{1,i}} A_1			& 1   \\ 
\hline 
= 0    & = 0    &  \Monome{X}{1} \cdot A_1                                  		 & A_0       						& 2   \\ 
\hline 
\ne  0    & = 0 & \multicolumn{2}{c|}{\text{no update}}												  		& 3   \\ 
\hline 
\end{array} 
\]
\caption{The update formulas depending on the defects.
$A$ denotes $N$ or $W$, which have the same update formula, 
or $\vectgr{u}_s = (u_{s,i})_i$, for $s=0,1$.
We denote by $A'$ the update of $A$, in order to distinguish its value at the beginning and at the end of a round.
}
\label{tbl:updates}
\end{table}

	\subsection{Proof of the algorithm}
\label{part:prta}

First we suppose that we never fall in the case of the {\em secondary loop} (lines $10$ to $21$). 
The easy part is to prove that if the interpolation condition is satisfied at round $j$, then it is also satisfied at round $j+1$. As a consequence since the interpolation condition is satisfied by construction at the beginning of  {\em Interpolation step},  it is also satisfied at the end of the algorithm for both pairs of polynomials.  

\Oubliettes{
\textcolor{green}{
\begin{lm}[Discrepancies]\label{lm:discr}
Assume that at round $k \leq j \leq n-1$, we have
$u_{0,i} = \Eval{N_0}{g_i} - \Eval{W_0}{y_i}$ and $u_{1,i} = \Eval{N_1}{g_i} - \Eval{W_1}{y_i}$ for all $1 \leq i \leq n$.
Let $(N'_0,W'_0)$ and $(N'_1,W'_1)$ be the polynomials obtained from $(N_0,W_0)$ and $(N_1,W_1)$ from any update described in {\em \bf Table~\ref{tbl:updates}} and let $\vec{u}'_0$ and $\vec{u}'_1$ be the discrepancies from the same update. 
Then $u'_{0,i} = \Eval{N'_0}{g_i} - \Eval{W'_0}{y_i}$ and $u'_{1,i} = \Eval{N'_1}{g_i} - \Eval{W'_1}{y_i}$ for all $1 \leq i \leq n$.
\end{lm}
}
}
\Oubliettes{
\textcolor{green}{
\begin{proof}
We only prove for updates of type $1$, that is:
\begin{enumerate}
  \item At round $j$, $N'_1 = \left( \Monome{X}{1} - \frac{\theta(u_{1,j})}{u_{1,j}} \Monome{X}{00} \right) \cdot N_1$ ;  $W'_1 = \left( \Monome{X}{1} - \frac{\theta(u_{1,j})}{u_{1,j}} \Monome{X}{00} \right) \cdot W_1$ and $\vec{u}'_1= \theta(\vec{u}_1) - \frac{\theta(u_{1,j})}{u_{1,j}} \cdot \vec{u}_1$:
For all $i$, we have
\begin{align}
u'_{1,i} 	&= \theta(u_{1,i}) - \frac{\theta(u_{1,j})}{u_{1,j}} \cdot u_{1,i}\\
			&=\left(\theta(\Eval{N_1}{g_i})-\frac{\theta(u_{1,j})}{u_{1,j}}\Eval{N_1}{g_i}\right)
				-\left( \theta(\Eval{W_1}{y_i})-\frac{\theta(u_{1,j})}{u_{1,j}}\Eval{W_1}{y_i}\right)\\
			&=\Eval{N'_1}{g_i} - \Eval{W'_1}{y_i}.
\end{align}
  \item  At round $j$, $N'_0 =  N_0 - \frac{u_{0,j}}{u_{1,j}} N_1$ ; $W'_0 =  W_0 - \frac{u_{0,j}}{u_{1,j}} W_1$ and $\vec{u}'_0= \theta(\vec{u}_0) - \frac{u_{0,j}}{u_{1,j}} \cdot \vec{u}_1$:
  For all $i$, we have
\begin{align}
u'_{0,i} 	&= u_{0,i} - \frac{u_{0,j}}{u_{1,j}} \cdot u_{1,i}\\
			&=\left(\Eval{N_0}{g_i}-\frac{u_{0,j}}{u_{1,j}} \Eval{N_1}{g_i}\right)
				-\left(\Eval{W_0}{y_i}-\frac{u_{0,j}}{u_{1,j}} \Eval{W_1}{y_i}\right)\\
			&=\Eval{N'_0}{g_i} - \Eval{W'_0}{y_i}.
\end{align}
\end{enumerate}
The proof is similar for the other update.
Thus, the quantity $u_{\ell,i}$ corresponds to the discrepancy $\Eval{N_\ell}{g_i}-\Eval{W_\ell}{y_i}$ for $\ell=0,1$ and for $1\leq i \leq n$.
\qed
\end{proof}
}
}

\begin{proposition}[Interpolation]\label{Prop:Interpolation}
Let $k+1 \leq j \leq n-1$ such that  $\Eval{W_0}{y_i} = \Eval{N_0}{g_i}$ and  $\Eval{W_1}{y_i} = \Eval{N_1}{g_i}$ for all $1 \leq i \leq j$.
Let $(N'_0,W'_0)$ and $(N'_1,W'_1)$ be the polynomials obtained from $(N_0,W_0)$ and $(N_1,W_1)$ from any update described in {\em \bf Table~\ref{tbl:updates}}. 
Then     $\Eval{W'_0}{y_i} = \Eval{N'_0}{g_i}$ and  $\Eval{W'_1}{y_i} = \Eval{N'_1}{g_i}$, for all $1 \leq i \leq j+1$.
\end{proposition}

\begin{proof}
We only prove for updates of type $1$, that is:
\begin{enumerate}
  \item At round $j$, $N'_1 = \left( \Monome{X}{1} - \frac{\theta(u_{1,j})}{u_{1,j}} \Monome{X}{00} \right) \cdot N_1$ and  $W'_1 = \left( \Monome{X}{1} - \frac{\theta(u_{1,j})}{u_{1,j}} \Monome{X}{00} \right) \cdot W_1$:

For all $1 \leq i \leq j$, we have $\Eval{N'_1}{g_i} = \theta(\Eval{N_1}{g_i}) -  \frac{\theta(u_{1,j})}{u_{1,j}}\Eval{N_1}{g_i}$ and $\Eval{W'_1}{y_i} = \theta(\Eval{W_1}{y_i}) -  \frac{\theta(u_{1,j})}{u_{1,j}}\Eval{W_1}{y_i}$. 
From this we only have to check the equality for $i=j+1$, since for $i \leq j$ this comes from the hypotheses of the theorem. 
For $i=j+1$ we have $u_{1,j+1} = \Eval{N_1}{g_{j+1}} - \Eval{W_1}{y_{j+1}}$ \Oubliettes{\textcolor{green}{(see {\bf Lemma~\ref{lm:discr}})}} and by reordering the terms 
\[
\Eval{N'_1}{g_{j+1}} = \frac{1}{u_{1,j+1}} \left( \Eval{N_1}{g_{j+1}} \theta(\Eval{W_1}{y_{j+1}})   -   \theta(\Eval{N_1}{g_{j+1}})\Eval{W_1}{y_{j+1}} \right).
\]
 We obtain the same value for $\Eval{W'_1}{y_{j+1}}$. 
  \item  At round $j$, $N'_0 =  N_0 - \frac{u_{0,j}}{u_{1,j}} N_1$ and $W'_0 =  W_0 - \frac{u_{0,j}}{u_{1,j}} W_1$:

For all $1 \leq i \leq j$,  the equality comes from the hypotheses of the theorem: $\Eval{N_0}{g_i} = \Eval{W_0}{y_i}$ and  
$\Eval{N_1}{g_{i}} = \Eval{W_1}{y_{i}}$. 

Now since for $i = j+1$ we have  $u_{1,j+1} = \Eval{N_1}{g_{j+1}} - \Eval{W_1}{y_{j+1}}$ and 
$u_{0,j+1} = \Eval{N_0}{g_{j+1}} - \Eval{W_0}{y_{j+1}}$ \Oubliettes{\textcolor{green}{(see {\bf Lemma~\ref{lm:discr}})}}, we obtain 
\[
\Eval{N'_0}{g_{j+1}} = \frac{1}{u_{1,j+1}} \left(  \Eval{W_0}{y_{j+1}} \Eval{N_1}{g_{j+1}} - \Eval{W_1}{y_{j+1}} \Eval{N_0}{g_{j+1}} \right).
\]
The same value is obtained for  $\Eval{W'_0}{y_{j+1}}$.

\end{enumerate}
For the other updates, the interpolation property is obviously satisfied.
\qed
\end{proof}

We have proved that our algorithm correctly interpolated the polynomials at every round in the {\em Interpolation step}.  
Now  we have to control the degrees to check that at least one of the pairs of polynomials 
satisfies the degree conditions (\ref{Eq:DegCond}).  

First we give an upper bound on the degrees of the polynomials in the algorithm.  
\begin{proposition}[Degree Control]\label{Prop:BSupDegre}
At the end of round $j$,  $k+1 \leq j \leq n$ of  the Interpolation step in  {\bf Algorithm \ref{algo:ta}}, 
the degrees of the polynomials satisfy:
\[
\begin{array}{ll}
 \deg(N_0) \leq k + \lfloor \frac{j-k}{2} \rfloor, & ~ \deg(W_0) \leq  \lfloor \frac{j-k+1}{2} \rfloor, \\
 \deg(N_1) \leq k - 1 + \lfloor \frac{j-k+1}{2} \rfloor, & ~  \deg(W_1) \leq  \lfloor \frac{j-k}{2} \rfloor.
\end{array}
\]
\end{proposition}

\begin{proof}
  The proof is made by induction. At the beginning of round $j=k+1$, $N_0$ and $N_1$ are respectively the annihilator of degree $k$ 
and the interpolating polynomial of degree $k-1$. Moreover  $W_0=0$  and $W_1 = X$. Therefore, 
by considering the updates in {\bf Table \ref{tbl:updates}} and the fact that the polynomials are 
swapped at the end of round $j$ (lines $37$ to $40$), at the end of the round 1, we have: 
\begin{itemize}
  \item $\deg(N_0) = k,~ \deg(W_0) = 1$, since polynomials are multiplied by an affine polynomial, increasing thus their degrees exactly by one.   
  \item $\deg(N_1) = k$, since it is the sum of a polynomial of degree exactly $k$ and a polynomial of degree strictly less than $k$. 
  \item $\deg(W_1) \leq 1$. There is no certainty on the exact degree of $W_1$ since it is the sum of $0$ and a constant, possibly $0$.  
\end{itemize}
Suppose that the property is true for some $k+1 \leq j \leq n-1$. There are two cases: 
\begin{enumerate}
  \item $j-k = 2u$ is even. 

 By hypothesis, at the beginning of round $j+1$ (corresponds to the end of round $j$), we have 
$\deg(N_0) \leq k+u$,  $\deg(W_0) \leq u$, $\deg(N_1) \leq k-1+u$,  $\deg(W_1) \leq u$. 
At the end of round $j+1$ we have therefore, 
\[
\begin{array}{ll}
  \deg(N_0) \leq k+u, &     ~\deg(W_0) \leq u +1,\\
 \deg(N_1) \leq k-1 +u+1, &  ~\deg(W_1) \leq u.
\end{array}
\] 
Since $j-k = 2u$ is even,   $\lfloor \frac{j+1-k}{2} \rfloor = u$, and  $\lfloor \frac{j+1-k +1}{2} \rfloor = u+1$.
  
\item $j-k = 2u +1$ is odd.  

By hypothesis, at the beginning of round $j+1$
$\deg(N_0) \leq k+u$,  $\deg(W_0) \leq u+1$, $\deg(N_1) \leq k+u$,  $\deg(W_1) \leq u$.
At the end of round $j+1$, 
\[
\begin{array}{ll}
 \deg(N_0) \leq k+u+1,    &     ~\deg(W_0) \leq u+1, \\
 \deg(N_1) \leq k-1 +u+1, &     ~\deg(W_0) \leq u+1.
\end{array}
\] 
Since $j-k = 2u+1$, we have: $\lfloor \frac{j+1-k +1}{2} \rfloor = \lfloor \frac{j+1-k}{2} \rfloor = u+1$.
\end{enumerate}

Now suppose that the upper bound on the degrees is true for some $k+1 \leq j \leq n-1$. Then the chosen updates show that it is still true for $j+1$. \qed
\end{proof}

So far this proposition gives upper bounds, but does not ascertain that at the end of the algorithm we will not fall into a degenerated case ($W_1 = 0$). To this end we will use the following proposition which shows that at every round at least one polynomial of every pair reaches the degree upper bound of the previous proposition. 

\begin{proposition}\label{Prop:NonDegenere}

At the end of loop $k+1 \leq j \leq n$ in  {\bf Algorithm \ref{algo:ta}}: 
\begin{itemize}
  \item If $j-k = 2u + 1$, then $\deg(N_1) = k+u$  and $\deg(W_0) = u+1$;
  \item If $j-k = 2u$, then $\deg(N_0) = k+u$ and $\deg(W_1) = u$.

\end{itemize}
\end{proposition}

\begin{proof}
Let $\mathcal{P}_{j}$ for all $k \leq j \leq n-1$ the property:  
\begin{itemize}
   \item  $\deg(N_1) = k+u, ~\deg(W_0) = u+1$, if $j = k + 2u+ 1$.
   \item  $\deg(N_0) = k+u, ~\deg(W_1) = u$, if $j=k+2u$.
\end{itemize}

From the initialization round in the proof of {\bf Proposition~\ref{Prop:BSupDegre}}, $\mathcal{P}_{k+1}$ is satisfied. 
Suppose now that  $\mathcal{P}_{k+2u + 1}$ is satisfied,  we show that 
 $\mathcal{P}_{k+2u + 2}$ is satisfied. 

Combining the  induction property and the upper bounds of {\bf Proposition~\ref{Prop:BSupDegre}},  
at the end of round  $j = k+2u+1$, we have:  
\[
\begin{array}{ll}
  \deg(N_0) \leq k+u, & ~  \deg(W_0) = u+1, \\ 
  \deg(N_1) =   k+u,&  ~ \deg(W_1) \leq u.
\end{array}
 \]
Hence at the beginning of round $j+1 = k + 2u +2$, the same bounds and equalities hold. 

Since $\deg(W_0) > \deg(W_1)$  the  updates in the loop show that 
 at the end of the round $k + 2u +2$,  $\deg(W_1) = u+1$. The polynomial  $N_0$ 
is obtained from $N_1$ by a left multiplication by an affine monic polynomial. Hence $\deg(N_0) = k+u+1$
at the end of round $k + 2u +2$. The upper bounds come from {\bf Proposition~\ref{Prop:BSupDegre}}. Hence
  $\mathcal{P}_{k+2u + 2}$ is satisfied. 

By using  the same arguments we show that if $\mathcal{P}_{k+2u}$ is satisfied then $\mathcal{P}_{k+2u+1}$ is also satisfied. 
Since we proved that  $\mathcal{P}_{k+1}$ was satisfied, by induction we proved that $\mathcal{P}_{j}$ is satisfied for all   
$k+1 \leq j \leq n$. \qed 
\end{proof}

A direct consequence of this proposition is that $W = 0$  can never occur for any pair of polynomials in the {\em Interpolation step}.
Namely, this would imply that the corresponding polynomial $N$ satisfies $\Eval{N}{g_i}=0$,  for all $1 \leq i \leq j$. However, 
$\deg(N) \leq k + \lfloor \frac{j-k}{2} \rfloor < j$, for $j \ge k+1$. Since the $g_i$'s are linearly independent by the hypotheses on the input of the algorithm,  this implies that $N=0$. 
But from  {\bf Proposition \ref{Prop:NonDegenere}}, it is not possible to have the two polynomials of the 
pair which do not reach the upper bound on the degree.  

This is the reason why the algorithm  returns the pair of polynomials of smallest degree. 
Now by combining all the previous results, we obtain: 
\begin{theorem}[Proof of the algorithm]
\label{Theo:AlgoWB}
Let 
    \begin{itemize}
    \item $K \hookrightarrow L$ be a cyclic field extension;
    \item $\theta$ be a generator of its automorphism group $\Auto_K(L)$;
      \item $\vectgr{g}=\VecteurLigneCoins{g_1}{g_n} \in L^n$ be $K$-linearly independent elements;
      \item $\vectgr{y}=\VecteurLigneCoins{y_1}{y_n} \in L^n$;
      \item an integer $k \leq n \in \N$
    \end{itemize} 
Then the pair $(N_1,W_1)$ returned by {\bf Algorithm \ref{algo:ta}} is a solution of 
$\LR(n, k, t=\lfloor (n-k)/2 \rfloor, \vectgr{g},  \vectgr{y})$. 
\end{theorem}

\begin{proof}
    {\bf Proposition~\ref{Prop:Interpolation}} shows that the returned pair satisfies the interpolation step for all $j = 1,\ldots,n$. 
Let $t = \lfloor (n-k)/2 \rfloor$. 
From {\bf Proposition~\ref{Prop:BSupDegre}}  the degrees of the returned pair $(N_1,W_1)$ satisfy 
\[
\begin{array}{l}
\deg(N_1) \leq  
\begin{cases}
 k -1 + t,& \mbox{~if~} n-k = 2t \\
 k+t,&      \mbox{~if~} n-k = 2t+1 
\end{cases} \\ 
 \deg(W_1) \leq t  
\end{array}
\]
and {\bf Proposition~\ref{Prop:NonDegenere}} ensures that $W_1 \ne 0$, therefore from (\ref{Eq:DegCond}) the degree conditions are satisfied. \qed
\end{proof}

Now we deal with the case of what occurs in the case where the  {\em secondary loop} (lines $10$ to $21$) is activated. 
In the case where at round $i$, 
$u_{1,i} = 0$ and $u_{0,i} \ne 0$, we search for the first position $i < s \leq n$ such that either $u_{1,s} \ne 0$ or 
 $u_{1,s} = 0 = u_{0,s}$. There are two cases :
\begin{itemize}
  \item Either there exists some  $i < s \leq n$ satisfying  either $u_{1,s} \ne 0$ or 
 $u_{1,s} = 0 = u_{0,s}$. In that case the positions $i$ and $s$ are exchanged. This corresponds to a permutation of one position along
all the input vectors and has no impact on the interpolation and degree conditions. 
\item Or such an $s$ does not exist and this means that the pair of polynomials $(N_1,W_1)$ at the beginning of round $i$ satisfies the interpolation conditions (the discrepancy vector $\vectgr{u}_1$ is equal to $\vectgr{0}$), and from {\bf Proposition~\ref{Prop:BSupDegre}}  and {\bf Proposition~\ref{Prop:NonDegenere}} the degree conditions are satisfied. 

\end{itemize}

	\subsection{Complexity of the decoding}
\label{part:cplx}

The complexity of the decoding procedure consists of adding
\begin{itemize}
\item The complexity of {\bf Algorithm \ref{algo:ta}} returning the $\theta$-polynomials $(N_1,W_1)$;
\item The complexity of the left Euclidean division of $N_1$ by $W_1$.
\end{itemize}

\subsubsection{Prerequisite}
Before studying the complexity of the {\bf Algorithm \ref{algo:ta}}, we give the complexity of elementary functions used.
We count the number of additions in $L$, multiplications in $L$, uses of $\theta$ and divisions in $L$.
For certain field extensions $K \hookrightarrow L$, an element of $L$ is represented by its coefficients in a $K$-basis of $L$, 
and the automorphism consists in permuting these coefficients.
So uses of $\theta$ do not need computations in these cases, for example in Kummer or cyclotomic extensions.

Let $A, B \in L[X;\theta]$ be $\theta$-polynomials of degrees $a$ and $b$, and let $x \in L$.
The complexity of arithmetic operations is given in the following table.

\bigskip

\begin{center}
\begin{tabular}{|c|c|c|c|c|}
\hline 
operation			& additions 			& multiplications 	& uses of $\theta$ 	& divisions 	\\ 
\hline 
$A+B$ 				& $1+\min(a, b)$ 		& $0$ 			& $0$			& $0$ 		\\ 
\hline 
$A \cdot B$ 		& $ab$ 				& $(1+a)(1+b)$ 	& $a(1+b)$ 		& $0$ 		\\ 
\hline 
$\Eval{A}{x}$ 		& $a$ 				& $a+1$ 			& $a$ 			& $0$ 		\\
\hline 
Euclidean division	& 					&				& 				&			\\
$A = B \cdot Q + R$	& $(a-b)b$			& $(a-b)(b+1)$		& $(a-b)(2b)$		& $(a-b)$		\\
\hline
\end{tabular} 
\end{center}

\bigskip

The multiplication algorithm is the naive one, but in our case it is optimal since we only  multiply by polynomials of degree $1$.

\subsubsection{Error correction}	

At the  {\em Initialization step}, we have to compute the annihilator and the interpolating polynomials.
They can be simultaneously computed with {\bf Algorithm~\ref{algo:AnnEtInt}}.
 
\begin{algorithm}[!h]
\caption{Annihilator and Interpolator polynomials}
\label{algo:AnnEtInt}
\begin{algorithmic}[1]
\ION \REQUIRE $g_1, \ldots, g_k \in L$ $K$-linearly independent
\IOB \REQUIRE $y_1, \ldots, y_k \in L$
\ION \ENSURE $\Ann(X),~\Int(X)$
\STATE $\Ann:= \Monome{X}{0}$	
\STATE $\Int:= 0$
\FOR {$1 \leq i \leq k$}
\STATE $\Int:=\Int + \frac{y_i-\Eval{\Int}{g_i}}{\Eval{\Ann}{g_i}} \cdot \Ann$
\STATE $\Ann:=(\Monome{X}{1}-\frac{\theta(\Eval{\Ann}{g_i})}{\Eval{\Ann}{g_i}}\Monome{X}{00}) \cdot \Ann$
\ENDFOR
\RETURN $\Ann,\Int$
\end{algorithmic}
\end{algorithm}

Additionally we  have to compute the discrepancy vectors: 
\[
\begin{array}{l}
\vectgr{u}_0 = (0,\ldots,0,\Eval{\mathcal{A}}{g_{k+1}},\ldots,\Eval{\mathcal{A}}{g_n} ) \\
\vectgr{u}_1 = (0,\ldots,0,\Eval{\mathcal{I}}{g_{k+1}} - y_{k+1},\ldots,\Eval{\mathcal{I}}{g_n} - y_n) 
\end{array}
\]

Therefore the complexity of the {\em Initialization step} is: 
\medskip 
\begin{center}
\begin{tabular}{|c|c|c|c|c|}
\hline 
	& additions 	& multiplications 	& uses of $\theta$ 	& divisions\\ 
\hline 
{\bf Algorithm~\ref{algo:AnnEtInt}} & $2k^2-2k$ 	& $2k^2-k$ 		& $1.5k^2-0.5k$  	& $2k$	\\ 
\hline 
$\vectgr{u}_0$ and $\vectgr{u}_1$   & $(2k-1)(n-k)$               & $(2k+1)(n-k)$           & $(k-1)(n-k)$                & 0  \\
\hline    
\end{tabular} 
\end{center}
\medskip


In the  {\em Interpolation step} the number of arithmetic operations can be upper bounded by considering only updates 
of type $1$.  The involved operations consist of:
\begin{enumerate}
    \item updates of $(N_0, W_0)$ and $(N_1,W_1)$ at round $j=k+i$: 

From {\bf Proposition~\ref{Prop:BSupDegre}} the degrees of the $\theta$-polynomials $N_0$, $N_1$, $W_0$ and $W_1$ are respectively less or equal to  
$k+\lfloor \frac{i-1}{2} \rfloor$, $k-1+\lfloor \frac{i}{2} \rfloor$, 
$\lfloor \frac{i}{2} \rfloor$ and $\lfloor \frac{i-1}{2} \rfloor$. 

    \item updates of the discrepancy vectors  $\vectgr{u}_0$ and $\vectgr{u}_1$.  
The discrepancy vectors  $\vectgr{u}_0$,  $\vectgr{u}_1$ at round $j+1$ in the 
loop can be obtained from  $\vectgr{u}_0$,  $\vectgr{u}_1$  at round $j$ by performing the same updates 
on the vectors than for the corresponding pairs of polynomials. For instance if we consider updates of type $1$: 
\[
   \begin{array}{lcl}
\vectgr{u'}_1 &\longleftarrow& \theta(\vectgr{u}_1) - \frac{\theta(u_{1,j})}{u_{1,j}} \vectgr{u}_1, \\
\vectgr{u'}_0 &\longleftarrow& \vectgr{u}_0 - \frac{u_{0,j}}{u_{1,j}} \vectgr{u}_1,
\end{array}
\]
where $\theta$ acts on $\vectgr{u}_1$ component by component. 
\end{enumerate}

Therefore, at round $j$, an upper bound on the complexities is: 
\medskip 
\begin{center}
\begin{tabular}{|c|c|c|c|c|}
\hline 
{}										& additions	& multiplications       & uses of $\theta$		& divisions\\ 
\hline
Up. $N_0$, $N_1$, $W_0$ and $W_1$		& $2j-1$	& $2j-1$ 				& $j+1$					& $2$	\\
\hline
Up. $\vectgr{u}_0$ and $\vectgr{u}_1$	& $2(n-j)$	& $2(n-j)$              & $ n-j  $				&    $0$	\\
\hline
Total 									& $2n-1$	& $2n-1$				& $n+1$					& $2$	\\
\hline
\end{tabular} 
\end{center}
\medskip

To obtain the full cost one has to sum the complexities for $j=k+1,\ldots,n$. Finally to complete the decoding complexity analysis  
it remains to evaluate the complexity of the final left Euclidean division of 
$N_1$ by $W_1$ which is: 
\medskip 
\begin{center}
\begin{tabular}{|c|c|c|c|c|}
\hline 
 {}					& additions	& multiplications        & uses of $\theta$		& divisions\\ 
\hline
Left Euclidean division 	& $(k-1)\frac{n-k}{2}$	& $(k-1)\frac{n-k}{2}$              & $ (n-k)(k-1) $			&    $0$	\\
\hline
\end{tabular} 
\end{center}
\medskip
Since the {\em secondary loop} consists only on tests and permutations, it has no effect on the complexity. 
All these evaluation lead to the following theorem: 

\begin{theorem}[Decoding complexity]
\label{Theo:DecWB}
The complexity of solving 
$\Dec(K,L,\theta,n,k,\lfloor \frac{n-k}{2} \rfloor , \vectgr{g},\vectgr{y})$ by using 
{\bf Algorithm~\ref{algo:ta}} is $O(n^2)$ operation in $L$. 
More precisely the number of different field operations is upper-bounded by:
\begin{itemize}
	\item $2n^2-2n					+(k-1)\frac{n-k}{2}$ 				additions in $L$,
	\item $2n^2-k					+(k-1)(\frac{n-k}{2})$ 			multiplications in $L$,
	\item $n^2+0.5k^2-2n+1.5k^2 	+(n-k)(k-1)$ 				        uses of $\theta$,
	\item $2n	$ 														divisions in $L$.
\end{itemize}
\end{theorem}

\begin{remark}
	Since $k(n-k) \leq \frac{n^2}{4}$, the total number of multiplication is upper-bounded by $2.125 n^2$.
\end{remark}

\subsubsection{Complexity of the errors and erasures correction}

Correcting errors and erasures corresponds to remove the erasures and then correcting the residual rank errors. 
As seen at section \ref{sssec:eff}, this corresponds to:
\begin{itemize}
  \item Removing the column erasures by either directly puncturing the columns or computing some Gaussian elimination and then puncturing;
  \item Removing row  erasures by computing $\mathcal{V}_r$ and evaluate it on $(n-s_c)$ elements of $L$;
  \item Decoding errors in a generalized Gabidulin code of length $n-s_c$ and dimension $k+s_r$;
  \item Divide the obtained solutions by  $\mathcal{V}_r$.
\end{itemize}

Overall a direct corollary of the previous theorem is 
\begin{corollary}
Given a generalized Gabidulin code of parameters $[n,k,d]$ and a received word $\vectgr{y}$ 
with  $s_r$ row erasures and $s_c$ column erasures and the maximal number of full-errors, 
the complexity of recovering the information polynomial $f$ is:
\begin{itemize}
\item $O(n s_c m)$ operations in $K$;
\item $O(n^2)$ multiplications in $L$.
\end{itemize}
\end{corollary}

	\subsection{Some improvements}
\label{sssec:amel}

In this section we present two ways to improve the complexity of {\bf Algorithm~\ref{algo:ta}}. 
In a first time, we design a variant without divisions in $L$. This can be of interest since 
although in finite fields, the complexity  of multiplication and division is roughly identical, 
this is  not necessarily the case in other fields.

Second, we use the property that the polynomials $N_0$ and $N_1$ lie by construction in the left module generated by 
the annihilator and the interpolating polynomials computed at the {\em Initialization step}. This enables to reduce significantly the cost
of the updates by updating polynomials of smaller degrees.

\subsubsection{A division-free variant}

As we will see in a next section by computing in integer rings of rational fields, it can be very interesting to 
process the algorithm without making divisions, so that when inputs are integer values, there is no fraction along the computation. 
This can be accomplished modifying {\bf Algorithm~\ref{algo:ta}} to avoid divisions (see Table~\ref{tbl:DFupdates}).

\begin{itemize}
\item Consider the polynomials $\mathcal{A},~\mathcal{I}$ returned by {\bf Algorithm~\ref{algo:DFAnnEtInt}} 
on the first $k$ positions of the input vectors of  {\bf Algorithm~\ref{algo:ta}}. We have
  \begin{itemize}
  \item $\mathcal{A} = \mu \Annul{<g_1, \ldots, g_k>}$.
  \item $\mathcal{I} = \lambda \Interpol{[g_1, \ldots, g_k]}{[y_1, \ldots, y_k]}$.
  \end{itemize}
  Both polynomials are computed without divisions. It is not difficult to  see that $\lambda \ne 0$ and $\mu \ne 0$. 
Therefore, by replacing lines  $4$ and $5$ of  {\bf Algorithm~\ref{algo:ta}} by 
$N_1(X) \longleftarrow \mathcal{I}$ and $W_1(X) \longleftarrow \lambda$, 
we obtain polynomials still satisfying the interpolation conditions and of the same degree.  

\item Concerning the interpolation step of   {\bf Algorithm~\ref{algo:ta}}, we modify the updates so that there is no more divisions by 
$u_1$: 

\begin{center}
\begin{table}[!h]
\[
\begin{array}{|c|c|r|l|c|}
\hline 
	u_{0,i}    
	& u_{1,i}    
	& A'_1                                                     		 
	& A'_0							
	& \text{type}\\ 
\hline 
	* 	
	& \neq 0 
	& \left(u_{1,i}\Monome{X}{1} - \theta(u_{1,i}) \Monome{X}{00} \right) \cdot A_1 
	& u_{1,i} A_0 - u_{0,i} A_1			
	& 1   \\ 
\hline 
	= 0    
	& = 0    
	&  \Monome{X}{1} \cdot A_1                                  		 
	& A_0       						
	& 2   \\ 
\hline 
	\ne  0    
	& = 0 
	& \multicolumn{2}{c|}{\text{no update}}												  		
	& 3   \\ 
\hline 
\end{array} 
\]
\label{tbl:DFupdates}
\caption{Division-free updates}
\end{table}
\end{center}

\end{itemize}

\begin{algorithm}[!h]
\caption{Annihilator and Interpolator polynomials (Division-free variant)}
\label{algo:DFAnnEtInt}
\begin{algorithmic}
\ION \REQUIRE $g_1, \ldots, g_k \in L$ $K$-linearly independent
\IOB \REQUIRE $y_1, \ldots, y_k \in L$
\ION \ENSURE $\Ann(X)$ such that $\Eval{\Ann}{x_i}=0$
\IOB \ENSURE $\Int(X)$ and $\lambda$ such that $\Eval{\Int}{x_i}=\lambda y_i$
\STATE $\Ann \longleftarrow \Monome{X}{0}$	
\STATE $\Int \longleftarrow 0$
\STATE $\lambda \longleftarrow 1$
\FOR {$1 \leq i \leq k$}
\STATE $\Int \longleftarrow \Eval{\Ann}{x_i} \cdot \Int + (\lambda y_i-\Eval{\Int}{x_i})\cdot \Ann$
\STATE $\lambda \longleftarrow \Eval{\Ann}{x_i} \cdot \lambda$
\STATE $\Ann \longleftarrow (\Eval{\Ann}{x_i}\Monome{X}{1}-\theta(\Eval{\Ann}{x_i})\Monome{X}{00}) \cdot \Ann$
\ENDFOR
\RETURN $\Ann,\Int,\lambda$
\end{algorithmic}
\end{algorithm}

Since $W_1$ is not monic anymore, this requires more multiplications. 
Concerning the Euclidean division  $ W_1 \backslash N_1$, since $W_1$ is not monic, this amounts to $k-1$ additional divisions, but 
these are  exact divisions.


Hence this   division-free variant of  {\bf Algorithm~\ref{algo:ta}} requires roughly $1.5$ times more 
multiplications in $L$.

\subsubsection{Polynomials with lower degree}

In the algorithm  $N_0$ and $N_1$ are updated using  additions and left-multiplications.
Therefore they lie in the left-module generated by $\mathcal{A}$ and $\mathcal{I}$. 
This implies that  at every round $k+1 \leq j \leq n$ 
they can be expressed under the form 
\[
N_i= P_i \cdot \Ann + Q_i \cdot \Int , ~i=0,1, 
\]
for some polynomials $P_i$ and $Q_i$, which are updated similarly to the corresponding polynomial $N_i$.  
Moreover, the polynomials $Q_i$ are initialized by $Q_1=W_1=\Monome{X}{0}$ and $Q_0=W_0=0$. Since the polynomials have the 
same initialization and the same update, they are equal, that is at every round we have: 
\[
N_i= P_i \cdot \Ann + W_i \cdot \Int , ~i=0,1. 
\]
If we replace the lines $25$ and $27$ in the algorithm with 
\[
\begin{array}{lcl}
P'_1 &\longleftarrow& (\Monome{X}{1} - \frac{\theta(u_{1,i})}{u_{1,i}} \Monome{X}{00}) \cdot P_1,\\
P'_0 &\longleftarrow& P_0 - \frac{u_{0,i}}{u_{1,i}} P_1,
\end{array}
\]
and lines $31$ and $33$ accordingly, we now  update $W_i$ and $P_i$, such that $\deg{P_i} = \deg(N_i) - k$.
The discrepancy vectors are updated as before. Hence the number of operations at round  $j$ is now upper bounded by  
 \begin{center}
\begin{tabular}{|c|c|c|c|c|}
\hline 
 {}					& additions	& multiplications        & uses of $\theta$		& divisions\\ 
\hline
 $\vectgr{u}_0$ and $\vectgr{u}_1$	& $2(n-j)$	& $2(n-j)$              & $ n - j +1  $			&    $0$	\\
\hline
Up. ($P_j$ and $W_j$)	& $2j-k-1$	&                  $2j-k-1$ 			& $j-k+1$				& $2$	\\
\hline
\end{tabular} 
\end{center}

The gain in arithmetic operations for the interpolation step is $k(n-k)$
The final division is modified and becomes:
\[
W_1 \backslash (P_1 \Ann) + \Int. 
\]
Since the polynomials in the division have the same degree than in the basic algorithm, the complexity is the same.
We need $k+1$ additional operations to modify the final division. However this effect is largely compensated by 
the fact that the complexity of computing the   $P_i$'s is much smaller than the complexity of computing the  $N_i$'s. 

\begin{theorem}[Complexity improvement]
If one consider using the polynomials $P_i$ rather than $N_i$ in  {\bf Algorithm~\ref{algo:ta}}, 
the number of arithmetic operations over the field $L$ is reduced by  $k(n-k)$.
\end{theorem}

	
\section{The case of number fields}
\label{sssec:compNF}

	When the code alphabet is an infinite field, like a number field, the
proposed decoding algorithm has the disadvantage that the bit-size of intermediate
coefficients  grows a lot, in such a way that the decoding algorithm
is not practical at all. A standard computer algebra way of
circumventing this problem is to perform computations modulo a large enough
prime. In the Section, we discuss this technique in the number field
case, and relate a generalized Gabidulin code built with integral
elements to its reduction modulo a prime, which turns out to be  a classical
Gabidulin code over a finite field.

Assuming that the receiver knows an \textit{a priori} bound on the size of the
message, or on the size of the error, it can apply this technique with
a large enough prime, to get the exact result over the number field,
only by doing computations modulo the chosen prime.




	\subsection{Basic algebraic number theory}
In this Section, we recall some definitions and properties about
number fields, their integer rings,  ideals and ramification.  We
refer the reader to~\cite[§4.8.1 and~§4.8.2]{cohen1993course}.

Let $\Q\hookrightarrow F$ be a number field of degree $[F:\Q]=m$, and
denote by $\OF$ its \emph{integer ring}.  The integer rings of number
fields are \emph{Dedekind rings}: a non zero ideal is prime if and
only if it is maximal.
A prime ideal of $\OF$ restricted to $\Z$ is a prime ideal $p\Z$ of $\Z$.
Conversely, let $p\Z$ be a prime ideal of $\Z$.  The ideal
$\mathfrak{p}$ generated by $p\Z$ in $F$ is generally not a prime
ideal, and we have the following decomposition into prime ideals:
\begin{align} 
\mathfrak{p}=\prod_{i=1}^q \mathfrak{p}_i^{e_i} 
\end{align} 
where the $\mathfrak{p}_i$'s are
the prime ideals of $\OF$ whose restriction to $\Z$ is $p\Z$.  The
exponent $e_i$ is called the \emph{ramification index}.  The
\emph{residue field} $\OF / \mathfrak{p}_i$ is an extension of
$\Z/p\Z$, its extension degree $f_i$ is called the \emph{residual
  degree}.  We say that $\mathfrak{p}_i$ is \emph{above}
$\mathfrak{p}$, and conversely that $\mathfrak{p}$ is \emph{below}
$\mathfrak{p}_i$,

We have the relation $ \sum_{i=1}^q f_i e_i = m $ and when
$\Q \hookrightarrow F$ is a Galois extension, we have that all the
indices $e_i$ are equal to the same $e$, and all residual degrees
$f_i$ are equal to the same $f$, with $qfe=m$.

We say that $p$ is \emph{ramified} if there is some $i$ such that $e_i > 1$. 
It happens only for a finite number of prime numbers $p$.
The number $p$ is \emph{inert} if it is unramified and $q=1$ (thus $f=m$).
Conversely, $p$ \emph{splits totally} in $\Z$ if it is unramified and $q=m$ (thus $e_i=f_i=1$).

Assume that $F=\Q[z]$ where $z$ is an algebraic integer. Let $T(X)\in\Q[X]$ be the minimal polynomial of $z$. 
Then for any prime which doesn't divide the index $[\OF:\Z[z]]$, 
we can obtain the prime decomposition of $\mathfrak{p}$ from the decomposition of $T(X)$ modulo $p$.
In particular, the prime $p$ is inert if $T(X)$ is irreducible modulo~$p$~\cite[Thm 4.8.13]{cohen1993course}. 

\begin{example}
	In the ring $\Z[i]$, 
	\begin{itemize}
		\item $2$ is ramified. Indeed, $2=(1+i)(1-i)$ but $i+1$ and $i-1=i(i+1)$ span the same ideal.
		\item $3$ is inert, since it is a prime number in $\Z[i]$. Thus, $\Z[i]/(3) \simeq \F_9$.
		\item $5$ splits totally. Indeed, $5=(2+i)(2-i)$ and these factors span two distinct ideals. Thus, $\Z[i]/(5) \simeq \F_5 \times \F_5$.
	\end{itemize}
\end{example}
Actually, to build relevant examples of generalized Gabidulin codes, we shall
need a more elaborate situation, with a base field
being already an extension of $\Q$. Let
$\Q\hookrightarrow K \hookrightarrow L$ be an extension of number
fields, $p\in \N$ a prime number, $\mathfrak p$ a prime ideal of $\OK$
above $p$, let $\OL$ be the ring of the integers of $L$, and
$\mathfrak P$ a prime ideal of $\OL$ above $\mathfrak p$, i.e.\
$\mathfrak p \in \mathfrak P$. Then $\OL/\mathfrak{P}$ is finite
degree extension of $\OK/\mathfrak{p}$. We can describe the Galois group of
$\OK/\mathfrak{p} \hookrightarrow \OL/\mathfrak{P}$ from the Galois
group $\Auto_K( L)$.  Recall that the \emph{decomposition group} of
$\mathfrak{P}$ is the following subgroup of the Galois group:
\[
D_{\mathfrak P}=\{\theta \in \Auto_K(L) : \theta(\mathfrak{P})=\mathfrak{P} \}
\]
whose cardinal is $ef$~\cite{cohen1993course}.  Consider the map
\begin{equation}\label{eq:psi}
\psi_{\mathfrak P} : 
\begin{array}[t]{rcl}
D_{\mathfrak P} &\rightarrow & \Auto_{\OK/\mathfrak p}(\OL/\mathfrak{P})\\
\theta & \mapsto& \overline{\theta}
\end{array}
\end{equation}
where  $\overline\theta$ is as follows:
\begin{equation}\label{eq:thetabar}
\overline\theta:
\begin{array}[t]{rcl}
 \OL/\mathfrak P &\rightarrow & \OL/\mathfrak P\\
x+\mathfrak{P} & \mapsto&\theta(x)+ \mathfrak{P}
\end{array}.
\end{equation}
Then $\psi_{\mathfrak P}$ is a morphism of groups with kernel
\[ I_{\mathfrak P}= \left\lbrace \theta \in D_{\mathfrak P} : \forall x \in \OL, \theta(x)-x \in \mathfrak{P} \right\rbrace \]
which is the \emph{inertia group} of $\mathfrak{P}$.

\begin{proposition}\cite[§6.2]{samuel1971theorie} \cite[§6.2]{samuel1971theorie}
\label{prop:thno}
	With the previous notation,
	\begin{enumerate}
		\item $\OL/\mathfrak{P}$ is a Galois extension of $\OK/\mathfrak{p}$, of degree $f$,
		\item $\psi_{\mathfrak P}$ is a surjective morphism
                  from $D_{\mathfrak P}$ to the Galois group $\Auto_{\OK/\mathfrak p}(\OL/\mathfrak{P})$,
		\item $I_{\mathfrak P}$ has cardinal $e$.
	\end{enumerate}
\end{proposition}

\subsection{Integral codes}

We now can give some definitions about codes with integer
coefficients.  Let $\Q\hookrightarrow K \hookrightarrow L$ be a number field and $\OL$
its integer ring (also called maximal order).  We also consider an
integral basis $\mathcal{B} = \VecteurLigneCoins{b_1}{b_\ell}$ of
$\OL$ (see~\cite[§4.4]{cohen1993course}). An integral basis is a basis
of the $\Z$-module $\OL$, and any element $x \in \OL$ can be uniquely
decomposed as
\[
 x = \sum_{i=1}^{\ell} x_i b_i, x_i \in \Z. 
\]
For any element $x \in L$, there exists $\lambda \in \Z$ such that
$\lambda x \in \OL$.  Let $\mathcal{C} \subset L^n$ be a code with
generating matrix $G \in \M_{k,n}(L)$.  Then $\mathcal C$ admits a
generating matrix $G' \in \M_{k,n}(\OL)$, with $\lambda\in \Z$ such
that $\lambda G=G'$.  For Gabidulin codes, we will furthermore consider
a support $\vectgr g$ of integral elements.
\begin{prop}
  Let $\Q\hookrightarrow K \hookrightarrow L$ be an extension of number
  fields, and $\OK$, $\OL$ the corresponding integers rings.  Let
  $\vecteurgras{g} \in L^n$ be a vector of $K$-linearly independent
  elements of $L$. Let $\theta \in \Auto_K(L)$, $0\leq k\leq n$, and
  $\mathcal C$ be the code $\mathcal C=\Gab_{\theta,k}(\vectgr
  g)$. Then there exists a support $\vecteurgras{g'} \in \OL^n$ such that $C=\Gab_{\theta,k}(\vectgr g')$.
\end{prop}
\begin{proof}
  Let $\lambda\in \Z$ such that $\lambda g_i\in \OL$,
  $i=1,\dots,n$. Such a $\lambda$ exists since we have a finite number
  of $g_i$'s. Then $\vecteurgras{g'}=\lambda\vecteurgras{g}\in \OL^n$
  is an integral support which defines the same code.\qed\end{proof}
\begin{definition}
Let $\Q\hookrightarrow K\hookrightarrow L$ be a number field and $\OL$
its integer ring (also called maximal order). Let $\mathcal C$ be a code in $L^n$.
Its associated \emph{integral code} is $\OC =\mathcal C \cap \OL^n$.
\end{definition}
\begin{definition}
  Let $Q\hookrightarrow K \hookrightarrow L$ be an extension of number
  fields, and let $C$ be Gabidulin code with  an integral
  support $\vecteurgras{g} = (g_1, \ldots, g_n) \in \OL^n$, its
  \emph{restricted code} is
\[ 	
	\mathcal{G} =
	\left \{ \VecteurLigneCoins{\Eval{f}{g_1}}{\Eval{f}{g_n}}~:~
	f \in \OL[X;\theta], \deg(f)<k 
	\right\}.
\]
\end{definition}

\begin{remark}
	The name of this code due to the coefficients of $f$, restricted to $\OL$.
	The restricted code is included in the integral code of the Gabidulin code.
	Nevertheless, the inclusion is strict. Consider for example a code whose support is 
	$(2, 2\alpha, 2\alpha^2, \ldots)$, and the information word $f=\frac{1}{2}X^0$.
	The corresponding codeword is $(1, \alpha, \alpha^2, \ldots)$, 
	which belongs to the integral code but not to the restricted code.
\end{remark}

\subsection{The intermediate growth of coefficients}


\begin{definition}
  Let $\Q \hookrightarrow K \hookrightarrow L$ be a number field
  extension of degree $\ell=[L:\Q]$, and let $\OK$ and $\OL$ be the
  corresponding integer rings, and
  $\mathcal{B} = \VecteurLigneCoins{b_1}{b_\ell}$ an integral basis of
  $\OL$. 
The \emph{size} of $x\in \OL$, with  $x$ uniquely written
\[
 x = \sum_{i=1}^{\ell} x_i b_i, x_i \in \Z,
\] is defined as
$ \vert x \vert = \log_2 (\max_i (\vert x_i \vert)) $, where
$\vert x_i \vert$ denotes the absolute value of $x_i \in \Z$.  The
\emph{size} of a polynomial in $\OL[X;\theta]$ or, of a vector in $L^n$, is the
maximal size of its coefficients.
\end{definition}
An observation is that even if the inputs of {\bf Algorithm
  \ref{algo:ta}} are small, and also if the output polynomial $f(X)$ is
small, the size of the intermediate $\theta$-polynomials $N(X)$ and
$W(X)$ can be quite large, even though their division $f(X)$ is small.
If we consider the division-free variant of the algorithm (see
algorithm \ref{algo:DFAnnEtInt} and table of updates
\ref{tbl:DFupdates}) the size of the $\theta$-polynomials $(N_i,W_i)$
is roughly doubled at every step.

\begin{example}
  We consider a generalized Gabidulin code with dimension $k=4$ and
  length $n=8$ over the cyclotomic extension of $\Q$ with degree
  $m=10$.  We choose an information polynomial with size $1$, and we
  add an error of size $2$ and rank $2$.  After the initialization
  step of the algorithm, polynomials $(N_0,W_0)$ and $(N_1,W_1)$ are
  of size $12$.  After the last step they have size $191$ (not very
  far from $2^4$ and $2^8$). In {\bf Table~\ref{tbl:tps-comp-L}}, we report  timings  for decoding codewords of codes of length up to $16$.

\end{example}

\subsection{Integer Gabidulin codes modulo a prime ideal}
In this Section, we consider codes defined on the integer ring $\OL$.
We define their reduction modulo an ideal of $\OL$,
and study the case of Gabidulin codes.
\begin{definition}
	\label{Defi:CodeQuotient}
	Let $\Q\hookrightarrow K \hookrightarrow  L$ be a number field, $\mathcal{C} \subset \OL^n$ be an integral code and  $\mathfrak P$ be an ideal of $\OL$.
	The reduction of the code modulo $\mathfrak P$ is 
	\[ \overline{\mathcal{C}} = \{ (\overline{c}_1, \dots, \overline{c}_n) : (c_1, \dots, c_n) \in \mathcal{C} \} \]
	where $\overline{x}$ denotes the reduction of $x$ modulo $\mathfrak P$.
\end{definition}
We want to study a Gabidulin code modulo $\mathfrak P$.  Under
 conditions given in the following theorem, this code is well defined and is a Gabidulin code
over a finite field.
\begin{theorem}
  Let $\Q \hookrightarrow K \hookrightarrow L$ be a number field extension and
  let $\Z$, $\OK$ and $\OL$ be the associated integer rings.  Let
  $\theta$ be a generator of the Galois group $\Auto_{K}(L)$. Let
  $\vectgr{g} = \VecteurLigneCoins{g_1}{g_n} \in \OL^n$ be a
  support of integral elements,  and $ \Gab_{\theta,k}(\vectgr{g}) $ be a generalized
  Gabidulin code with  support
  $\vectgr{g}$, whose associated
  restricted Gabidulin code is 
	\[ 
		\OGab_{\theta,k}(\vectgr{g}) = \left \{ 
		\VecteurLigneCoins{\Eval{f}{g_1}}{\Eval{f}{g_n}} :
		f \in \OL[X;\theta], \deg(f)<k 
		\right\}.
	\]
	Finally, let $\mathfrak{P}$ be a prime ideal of $\OL$, 
	and let $\mathfrak{p}\subset \OK$ denote the ideal below $\mathfrak{P}$.
If the following conditions hold
	\begin{enumerate}
		\item $\theta(\mathfrak{P})=\mathfrak{P}$,
		\item $\overline{g_1}, \ldots, \overline{g_n}$ are
                  $\OK/\mathfrak{p}$-linearly independent,
	\end{enumerate}
then the following code
	\[ 
		\overline{\mathcal{G}} =\left \{
		\VecteurLigneCoins{\Eval{\overline{f}}{\overline{g_1}}}{\Eval{\overline{f}}{\overline{g_n}}} :
		\overline{f} \in (\OL/\mathfrak{P})[X;\overline{\theta}], \deg(\overline{f})<k 
		\right\},
	\]
	where $\overline{\theta}$ is defined in Eq.~\ref{eq:thetabar},
        is a classical Gabidulin code
        $ \Gab_{\overline\theta,k}(\overline{\vectgr g}) $ defined
        using the extension of finite fields
        $\OK/\mathfrak{p}\hookrightarrow\OL/\mathfrak P$, and the
        automorphism $\overline\theta$. Furthermore, any
        $\overline c\in \overline{\mathcal{G}} $ is the reduction
        modulo $\mathfrak P$ of a codeword $c\in\OGab_{\theta,k}(\vectgr{g}) $.
\end{theorem}

\begin{proof}
  Condition 2 makes $\overline{\vectgr g}$ a valid support of linearly
  elements in $\OL/\mathfrak P$. Then $\overline\theta$ is a generator
  element of the Galois Group of the extension of finite fields
  $\OK/\mathfrak{p}\hookrightarrow\OL/\mathfrak P$, from {\bf
    Proposition~\ref{prop:thno}}. Actually $\overline\theta $ is a
  power of the Frobenius automorphism $x\mapsto x^q$, with
  $q=\size{\OK/\mathfrak{p}}$.\qed
\end{proof}

\begin{remark}
	To get $\overline{g_1}, \ldots, \overline{g_n}$ linearly independent over $\OK/\mathfrak{p}$,
	we need $n \leq [\OL/\mathfrak P:\OK/\mathfrak{p}]=f$.
	Usually, codes are designed with $n=m$, or with $n$ close to $m$.
	Thus, inert primes (for which $f=m$) are of a particular interest.
\end{remark}

\subsection{Decoding using a prime ideal}

We exhibits a simple link between decoding a generalized Gabidulin
code and decoding its reduction modulo a prime. First we provide a Lemma.
\begin{lemma}\label{imm:b}
	Let $\vectgr{e}=(e_1,\ldots,e_n) \in \mathcal{O}_L$ be a vector of $K$-rank $t$. 
	Then the $\mathcal{O}_K/\mathfrak{p}$-rank of  $\overline{\vectgr{e}}=\vectgr{e} \pmod{\mathfrak{P}}$ is at most $ t$.
\end{lemma}
\begin{proof}
  There exists a monic $\theta$-polynomial $E$ of degree $t$ which
  vanishes on the $e_i$'s.  Therefore
  $\Eval{\overline{E}}{\overline{e_i}} = \overline{\Eval{E}{e_i}}=0$
  and $\overline E$ is non zero, since $E$ is monic.  From our
  definition of rank metric with annihilator polynomials, this implies
  that $\overline{\vectgr{e}}$ has rank less than or equal to $ t$.\qed
\end{proof}

\begin{theorem}\label{imm:a}
  Let $\mathcal{G}$ be an restricted generalized Gabidulin code,
  $\mathfrak{p}$ be an inert prime of $\OK$, such that
  $\overline{\mathcal{G}}$ is a generalized Gabidulin code.  Suppose
  that $N,W \in \OL[X;\theta]$ is solution to
  $\LR(n,k,t,\vectgr{g}, \vectgr{y}=\VecteurLigneCoins{y_1}{y_n} \in
  \OL^n)$ (see {\bf Definition~\ref{def:LRP}}).  Let
  $\overline{N},\overline{W} \in \left( \OL/\mathfrak{P} \right)
  [X;\overline{\theta}]$ be the reduction of the $\theta$-polynomials
  $N,W$ modulo $\mathfrak{P}$.  Then $(\overline{N},\overline{W})$ is
  a solution to
  $\LR(n,k,t,\overline{\vectgr{g}},\overline{\vectgr{y}})$.
\end{theorem}

\begin{proof}
By hypothesis $\overline{\vectgr{g}}$ is formed with linearly independent elements, 
and since the conditions on the degrees of the polynomials remain unchanged by taking the polynomial modulo $\mathfrak{P}$, 
it is thus sufficient to check that 
\[
\begin{array}{l}
\overline{\Eval{W}{y_i}} =  \Eval{\overline{W}}{\overline{y_i}} \ForallInt{1}{i}{n}, \\
\overline{\Eval{N}{g_i}} =  \Eval{\overline{N}}{\overline{g_i}} \ForallInt{1}{i}{n}.
\end{array}
\]
This  is immediate from the definition of $\overline{\theta}$ in {\bf Definition~\ref{Defi:CodeQuotient}}
and from the fact that the operation of taking modulo $\mathfrak{P}$ is a ring morphism.\qed
\end{proof}
Suppose that one receives the vector
$ \vectgr{y} = \Eval{f}{\vectgr{g}} + \vectgr{e}$, where
$\Eval{f}{\vectgr{g}} \in \mathcal{G}$ such that the coefficients of
$f(X)$ are taken in the integer ring $\OL$, and $\vectgr{e}$, such that
$\rang{\vectgr{e}} \leq t \leq \lfloor (n-k)/2 \rfloor$, is also
formed with elements of $\OL$.  Then
\[
	\overline{\vectgr{y}} = \Eval{\overline{f}}{\overline{\vectgr{g}}} + \overline{\vectgr{e}},
\]
where
$ \Eval{\overline{f}}{\overline{\vectgr{g}}} \in
\overline{\mathcal{G}}$, and
$\rang(\overline{\vectgr{e}}) \leq t \leq \lfloor (n-k)/2 \rfloor$.
In the case where $\overline{\mathcal{G}}$ is a generalized Gabidulin
code over $\OL/\mathfrak{P}$, then $\overline{f}(X)$ can be recovered by
any algorithm solving the linear reconstruction problem like the {\bf
  Algorithm~\ref{algo:ta}}.  The same is true if we consider erasures.

With this method we limit the growth of intermediate values and the
bit-complexity of the decoding procedure, since all computations are
completed in the finite field $\OL/\mathfrak{P}$ which is isomorphic
to $\F_{q^m}$, where $q=\size{\OK/\mathfrak p}$.
Recovering $\overline{f}$ requires $O(n^2)$ operations in
$\F_{q^m}$. 


\begin{remark}
	Cyclotomic and Kummer extensions are well adapted to this purpose. 
	Indeed, in cyclotomic extensions, there are many prime numbers of $\Z$ that give inert ideals.
	Conversely, considering Kummer extensions, many prime numbers of $\Z$ splits totally in the first extension, 
	and their prime factors remain inert in the second extension.	
\end{remark}

\section{Example and timings}
	In this section, we provide a full example to illustrate the decoding
algorithm of generalised Gabidulin code over a number field.   All steps are
detailled.  This example includes the reduction of the code modulo a prime, the
correction of row and column erasures in the network coding model, the
reconstruction algorithm ({\bf Algorithm~\ref{algo:ta}} and the final
division).  Then, we present  timings for
decoding a code with integer coefficients, when decoding over $\OL$
(i.e. in the field $L$), and when decoding modulo an inert prime ideal.

\subsection{A full example}
In this example, we consider the extension 
\[ K=\Q \hookrightarrow L=K[\alpha] = \Q[Y] / (1+Y+\cdots+Y^6)\]
provided with the $K$-automorphism $\theta$ defined by
$\theta : \alpha \mapsto \alpha^3$.  The family
$(1, \alpha, \ldots, \alpha^5) $ is a $K$-basis of $L$.

We consider a code of length $n=6$ and dimension $k=2$.
The information word has the form 
$f(X)=f_0\Monome{X}{00}+f_1\Monome{X}{1}, f_i \in L$,
where the $f_i$'s have the form 
$f_i = f_{i,0} + \cdots + f_{i,5} \alpha^5$. Furthermore, suppose
that only small message are encoded, i.e.\ $ f_{i,j} \in \{0, 1\}$,
and that this fact is known at the receiving end.

Since the coefficients of the information polynomial are small and belong to the integer ring of $L$,
we can reduce the code modulo an inert prime ideal such as $3\OL$.
Moreover, the possibles values $\set{0,1}$ of the coefficients $f_{i,j}$ are distinct modulo $3\OL$,
thus knowing $f$ modulo $3\OL$ enables to know $f$ in $L$.

This example takes place in the network coding erasure model.  The
support $\vectgr g=(g_1,\dots,g_6)$ of the code and the received word
$\vectgr y=(y_1,\dots,y_6)$ are the following :
\[ 
\begin{array}{ll}
	g_1=1,	\qquad		& y_1 =\alpha^5 + \alpha^3 - \alpha^2 + 2\alpha + 2\\
	g_2=\alpha,			& y_2 =\alpha^5 - \alpha^4 + \alpha^3 + \alpha^2 - 1\\
	g_3=\alpha^2,		& y_3 =-2\alpha^5 + 4\alpha^4 + \alpha^2 - 2\alpha\\
	g_4=\alpha^3,		& y_4 =-\alpha^5 + 2\alpha^4 + \alpha^3 - \alpha^2 + 3\\
	g_5=\alpha^4,		& y_5 =-2\alpha^5 - 2\alpha^2\\
	g_6=\alpha^5,		& y_6 =-\alpha^5 - \alpha^4 + \alpha^3 - 2\alpha^2 - \alpha + 2\\
\end{array}
\]
and the receiver also knows the following matrices :
\[ \matrice{A_r}=\begin{pmatrix} 1\\-1\\0\\1\\1\\-1 \end{pmatrix}
\text{ and }
\matrice{B_c}=\begin{pmatrix} 1&0&-1&0&0&1 \end{pmatrix} \]
that describe erasures in the network coding model.

We first reduce the $g_i$'s, the $y_i$'s and the matrices $A_r$ and $B_c$ modulo $3\OL$.
The support of the code and the received word become the following:
\[
\begin{array}{ll}
g_1=1,\qquad 				&y_1=\overline{\alpha}^5 + \overline{\alpha}^3 + 2\overline{\alpha}^2 + 2\overline{\alpha} + 2\\
g_2=\overline{\alpha},		&y_2=\overline{\alpha}^5 + 2\overline{\alpha}^4 + \overline{\alpha}^3 + \overline{\alpha}^2 + 2\\
g_3=\overline{\alpha}^2,	&y_3=\overline{\alpha}^5 + \overline{\alpha}^4 + \overline{\alpha}^2 + \overline{\alpha}\\
g_4=\overline{\alpha}^3,	&y_4=2\overline{\alpha}^5 + 2\overline{\alpha}^4 + \overline{\alpha}^3 + 2\overline{\alpha}^2\\
g_5=\overline{\alpha}^4,	&y_5=\overline{\alpha}^5 + \overline{\alpha}^2\\
g_6=\overline{\alpha}^5,	&y_6=2\overline{\alpha}^5 + 2\overline{\alpha}^4 + \overline{\alpha}^3 + \overline{\alpha}^2 + 2\overline{\alpha} + 2\\
\end{array}
\]
Then, we correct column erasures.
We notice that the column operations 
$C_1 \leftarrow C_1-C_6$
and
$C_3 \leftarrow C_3+C_6$
give a column reduced echelon form (but with pivots on the right).
We do the same operation on $\vectgr{g}$ and $\vectgr{y}$ 
and remove their last component (which contains all  erasures) :
\[ 
\begin{array}{ll}
	g_1=2\overline{\alpha}^5 + 1, \qquad	
	& y_1 =2\overline{\alpha}^5 + \overline{\alpha}^4 + \overline{\alpha}^2\\
	g_2=\overline{\alpha},		
	& y_2 =\overline{\alpha}^5 + 2\overline{\alpha}^4 + \overline{\alpha}^3 + \overline{\alpha}^2 + 2\\
	g_3=\overline{\alpha}^5 + \overline{\alpha}^2,	
	& y_3 =\overline{\alpha}^3 + 2\overline{\alpha}^2 + 2\\
	g_4=\overline{\alpha}^3,		
	& y_4 =2\overline{\alpha}^5 + 2\overline{\alpha}^4 + \overline{\alpha}^3 + 2\overline{\alpha}^2\\
	g_5=\overline{\alpha}^4,		
	& y_5 =\overline{\alpha}^5 + \overline{\alpha}^2\\
\end{array}
\]
Then, we correct row erasures.
The column matrix $\matrice{A_r}$ modulo $3\OL$ correspond to the element 
$2\overline{\alpha}^5 + \overline{\alpha}^4 + \overline{\alpha}^3 + 2\overline{\alpha} + 1 \in \OL$.
We compute the annihilator polynomial $\mathcal V_r(X)$ of the columns of $\matrice{A_r}\bmod 3\OL$:
\[\mathcal V_r(X)= (\overline{\alpha}^5 + \overline{\alpha}^2)X^0 +(1)X^1 \] and
evaluate this polynomial over the $y_i$'s.  We are now looking for
$F(X)=\mathcal V_r (X)\cdot f(X)$ instead of $f(X)$.  The evaluations
$z_i$, $i=1,\dots,5$  of $F(X)$ are:
\[ 
\begin{array}{ll}
	g_1=2\overline{\alpha}^5 + 1,\qquad	& 
	z_1 =2\overline{\alpha}^5 + 2\overline{\alpha}^4 + 2\overline{\alpha}^2 + 1\\
	g_2=\overline{\alpha},		& 
	z_2 =2\overline{\alpha}^5 + \overline{\alpha}^4 + \overline{\alpha}^3 + 2\overline{\alpha}^2 + 2\overline{\alpha} + 1\\
	g_3=\overline{\alpha}^5 + \overline{\alpha}^2,	& 
	z_3 =\overline{\alpha}^5 + \overline{\alpha}^3 + \overline{\alpha}^2 + 2\overline{\alpha} + 2\\
	g_4=\overline{\alpha}^3,		& 
	z_4 =2\overline{\alpha}^5 + \overline{\alpha}^4 + \overline{\alpha}^3 + 2\overline{\alpha}^2 + 2\overline{\alpha}\\
	g_5=\overline{\alpha}^4,		& 
	z_5 =2\overline{\alpha}^5 + 2\overline{\alpha}^2 + 1\\
\end{array}
\]
We now have to solve the reconstruction problem to get $N$ and $W$.
We now have a code of dimension $3$ and length $5$, no more erasures,
and rank errors only, and we can apply our Welch-Berlekamp
algorithm. The initialisation will consider the first three
evaluations and there will be $2$ iterations.  We use the
division-free variant for this example.

At  initialisation, we compute the following polynomials.
\[ \begin{array}{ll}
N_0(X)=	&(\overline{\alpha}^5 + \overline{\alpha}^4 + \overline{\alpha}^3 + 2\overline{\alpha}^2 + 2\overline{\alpha} + 2)X^0
		+(\overline{\alpha}^5 + 2\overline{\alpha}^2 + 2\overline{\alpha} + 2)X^1\\
		&+(\overline{\alpha}^4 + 2\overline{\alpha} + 1)X^2
		+(1)X^3
\\
W_0(X)=	&0
\\	
N_1(X)=	&(2\overline{\alpha}^5 + \overline{\alpha}^4 + \overline{\alpha}^2 + 2\overline{\alpha})X^0
		+(2\overline{\alpha} + 2)X^1\\
		&+(2\overline{\alpha}^5 + 2\overline{\alpha}^4 + 2\overline{\alpha}^3 + 2\overline{\alpha} + 1)X^2
\\														
W_1(X)=	&(1)X^0
\\
\end{array} \]
We also initialise the discrepancies:
\begin{align*}
	\vectgr{u}_0&=(0,0,0,2\overline{\alpha} + 2, \overline{\alpha}^5 + \overline{\alpha}^4 + \overline{\alpha}^3 + 2\overline{\alpha}^2 + \overline{\alpha} + 2)\\
	\vectgr{u}_1&=(0,0,0,\overline{\alpha}^5 + 2\overline{\alpha}^4 + 2\overline{\alpha}^3 + \overline{\alpha}^2 + 2\overline{\alpha} + 2,  2\overline{\alpha}^5 + \overline{\alpha}^3)
\end{align*}
We begin the first iteration by extracting discrepancies:
\[ \begin{array}{ll}
u_{0,4} &=N_0(g_4)-W_0(z_4)\\
    &=2\overline{\alpha} + 2\\
u_{1,4} &=N_1(g_4)-W_1(z_4)\\
    &=\overline{\alpha}^5 + 2\overline{\alpha}^4 + 2\overline{\alpha}^3 + \overline{\alpha}^2 + 2\overline{\alpha} + 2\\
\end{array} \]
Then we update the polynomials by the following formulae:
\[ \begin{array}{ll}
N'_0(X)=&(X^1-\frac{\theta(u_1)}{u_1}X^0) \times N_1(X)\\
W'_0(X)=&(X^1-\frac{\theta(u_1)}{u_1}X^0) \times W_1(X)\\
N'_1(X)=&N_0-\frac{u_0}{u_1}N_1(X)\\
W'_1(X)=&W_0-\frac{u_0}{u_1}W_1(X).
\end{array} \]
 We schitch them to get:
\[ \begin{array}{ll}
N_0(X)=N'_1(X)=	&(2\overline{\alpha}^5 + 2\overline{\alpha}^4 + \overline{\alpha} + 1)X^0
			+(2\overline{\alpha}^5 + \overline{\alpha}^4 + \overline{\alpha}^3 + 2\overline{\alpha}^2)X^1\\
			&+(\overline{\alpha}^5 + 2\overline{\alpha}^4 + \overline{\alpha}^3 + \overline{\alpha}^2 + 2\overline{\alpha})X^2
			+(2\overline{\alpha}^5 + 2\overline{\alpha}^3 + 2\overline{\alpha}^2 + 2\overline{\alpha} + 1)X^3\\
W_0(X)=W'_1(X)=	&(2\overline{\alpha}^5 + \overline{\alpha}^4 + 2\overline{\alpha}^3 + 2\overline{\alpha} + 1)X^0
			+(1)X^1\\
N_1(X)=N'_0(X)=	&(2\overline{\alpha}^3 + \overline{\alpha}^2 + 2\overline{\alpha} + 2)X^0
			+(\overline{\alpha}^5 + 2\overline{\alpha}^4 + 2\overline{\alpha}^3 + 2\overline{\alpha}^2 + \overline{\alpha})X^1\\
			&+(2\overline{\alpha}^4 + \overline{\alpha}^3 + 2\overline{\alpha})X^2
			+(1)X^3\\
W_1(X)=W'_0(X)=	&(2\overline{\alpha}^5 + 2\overline{\alpha} + 1)X^0\\
\end{array} \]
We also update and switch discrepancies, to get:
\[ \begin{array}{llllll}
	\vectgr{u}_0&=(0,&0,&0,&0, &2\overline{\alpha}^5 + 2\overline{\alpha}^4 + 2\overline{\alpha}^3 + 2\overline{\alpha} + 1)\\
	\vectgr{u}_1&=(0,&0,&0,&0, &2\overline{\alpha}^5 + 2\overline{\alpha}^4 + 2\overline{\alpha}^3 + \overline{\alpha}^2 + 2\overline{\alpha} + 1)\\
\end{array} \]
During the second iteration, we get:
\[ \begin{array}{ll}
u_{0,5} &=N_0(g_4)-W_0(g_5)\\
    &=2\overline{\alpha}^5 + 2\overline{\alpha}^4 + 2\overline{\alpha}^3 + 2\overline{\alpha} + 1\\
u_{1,5} &=N_1(g_4)-W_1(z_5)\\
    &=2\overline{\alpha}^5 + 2\overline{\alpha}^4 + 2\overline{\alpha}^3 + \overline{\alpha}^2 + 2\overline{\alpha} + 1\\
\end{array} \]
and 
\[ \begin{array}{ll}
N_0(X)=N'_1(X)=	&(2\overline{\alpha}^4 + \overline{\alpha}^3 + 2\overline{\alpha}^2 + \overline{\alpha} + 1)X^0
			+(2\overline{\alpha}^5 + \overline{\alpha}^4 + \overline{\alpha}^2 + 2)X^1\\
			&+(\overline{\alpha}^4 + \overline{\alpha}^3 + \overline{\alpha}^2 + 1)X^2
			+(2\overline{\alpha}^5)X^3\\
			&+(1)X^4\\
W_0(X)=W'_1(X)=	&(2\overline{\alpha}^4 + 2\overline{\alpha}^3 + 2\overline{\alpha}^2 + 2\overline{\alpha} + 1)X^0
			+(2\overline{\alpha}^3 + 2\overline{\alpha} + 1)X^1\\
N_1(X)=N'_0(X)=	&(2\overline{\alpha}^5 + 2\overline{\alpha}^3)X^0
			+(\overline{\alpha}^5 + \overline{\alpha}^3 + \overline{\alpha}^2 + \overline{\alpha} + 2)X^1\\
			&+(\overline{\alpha}^5 + 2\overline{\alpha}^4 + 2\overline{\alpha}^3 + 2\overline{\alpha}^2 + \overline{\alpha} + 1)X^2
			+(\overline{\alpha}^3)X^3\\
W_1(X)=W'_0(X)=	&(\overline{\alpha}^4 + \overline{\alpha}^3 + 2\overline{\alpha}^2 + \overline{\alpha} + 1)X^0
			+(1)X^1\\
\end{array} \]
Updating and switching discrepancies, we get
\[ \begin{array}{llllll}
	\vectgr{u}_0&=(0,&0,&0,&0, &0)\\
	\vectgr{u}_1&=(0,&0,&0,&0, &0)\\
\end{array} \]
which satisy the interpolation condition.
Then, the polynomials
\[ \begin{array}{ll}
N(X)=N_1(X)=	&(2\overline{\alpha}^5 + 2\overline{\alpha}^3)X^0
		+(\overline{\alpha}^5 + \overline{\alpha}^3 + \overline{\alpha}^2 + \overline{\alpha} + 2)X^1\\
		&+(\overline{\alpha}^5 + 2\overline{\alpha}^4 + 2\overline{\alpha}^3 + 2\overline{\alpha}^2 + \overline{\alpha} + 1)X^2
		+(\overline{\alpha}^3)X^3\\
W(X)=W_1(X)=	&(\overline{\alpha}^4 + \overline{\alpha}^3 + 2\overline{\alpha}^2 + \overline{\alpha} + 1)X^0
		+(1)X^1\\
\end{array} \]
are solution of the reconstruction problem.
We divide them to get 
\[ \begin{array}{ll}
F(X)=&(\overline{\alpha}^4 + 1)X^0
+(2\overline{\alpha}^5 + 2\overline{\alpha}^4 + 2\overline{\alpha}^2 + 2\overline{\alpha})X^1
+(\overline{\alpha})X^2,\\
\end{array} \]
then, we divide it by $\mathcal V_r(X)$ to recover
\[ \begin{array}{ll}
f(X)=&(\overline{\alpha}^2)X^0
+(\overline{\alpha}^5)X^1.\\
\end{array} \]
Thus, since the possible $f_{i,j}$ are $0$ or $1$ modulo $\mathfrak{P}$,
we have recovered the information polynomial $f(X)=\alpha^2X^0+\alpha^5X^1$.


\subsection{Timings}
For timings measurements, our generalized Gabidulin codes are constructed over a
cyclotomic extension
$\Q \hookrightarrow L=\Q[\alpha]=\Q[Y]/(1+Y+\cdots+Y^{m-1})$, for
small values of $m$, see Table~\ref{table:cyclofields}.  The
information words are on the form
$f(X)=\sum_{i=0}^{k-1}f_i\Monome{X}{i}$.  The error is constructed by making the
product of a vector $(e_1, \ldots, e_t)$ by a matrix of $t$ rows and
$n$ columns over $K$.  In order to have small coefficients, the
$f_i$'s and the $e_i$'s are on the form $\sum_{j=1}^{m} x_{j}\alpha^j$
with $x_j \in \{0;1\}$ and the expanded coefficients are chosen in
$\{-1; 0; 1\}$.

First, computations are done in $\OL$, since all coefficients are
integral.  Then, the received words are reduced modulo an inert prime
ideal $\mathfrak{P}$.  These reduced received words are decoded over
$\OL/\mathfrak{P}$.  This ideal is generated by the smallest prime
number of $\Z$ inert in $\OL$.  The time required for computation are
respectively presented in {\bf Tables~\ref{tbl:tps-comp-L}
  and~\ref{tbl:tps-comp-OL}}.

The algorithm has been writted in Magma V2.20-9.
The machine has 24 processors intel xeon X5690, 96 gigas of RAM, 
is 3.47gHz clocked, and has distribution ubuntu 14.
Time computation are obained with the Cputime function.
It corresponds to the time required to $50$ decodings.

\begin{table}
\centering
\begin{tabular}{|c|c|c|c|c|c|c|c|c|}
\hline
length $n$ of the code		&$4$	&$6$	&$8$	&$10$	&$12$	&$14$	&$16$	\\
\hline
degree $[\Q[\alpha]:\Q]$	&$4$	&$6$	&$10$	&$10$	&$12$	&$16$	&$16$	\\
\hline
prime ideal $\mathfrak{P}$	&$2\OL$	&$3\OL$	&$2\OL$	&$2\OL$	&$2\OL$	&$3\OL$	&$3\OL$	\\
\hline
\end{tabular}
\caption{Cyclotomic extensions used in our timings}\label{table:cyclofields}
\end{table}

\begin{table}
\centering
\begin{tabular}{|c|c|c|c|c|c|c|c|c|}
\hline 
$n \backslash k	$&$ 2 $&$ 4 $&$ 6 $&$ 8 $&$ 10 $&$ 12 $&$ 14 $&$ 16 $\\ 
\hline 
$4 				$&$ 0.07 $&$ 0.08 $&$ {} $&$ {} $&$ {} $&$ {} $&$ {} $&$ {} $\\ 
\hline 
$6 				$&$ 0.23 $&$ 0.24 $&$ 0.21 $&$ {} $&$ {} $&$ {} $&$ {} $&$ {} $\\ 
\hline 
$8 				$&$ 0.84 $&$ 0.85 $&$ 0.87 $&$ 0.85 $&$ {} $&$ {} $&$ {} $&$ {} $\\ 
\hline 
$10 			$&$ 1.78 $&$ 1.91 $&$ 2.11 $&$ 2.32 $&$ 2.51 $&$ {} $&$ {} $&$ {} $\\ 
\hline 
$12 			$&$ 10.59 $&$ 11.00 $&$ 12.52 $&$ 15.15 $&$ 17.72 $&$ 20.32 $&$  {} $&$ {} $\\ 
\hline 
$14 			$&$ 215.21 $&$ 196.06 $&$ 202.07 $&$ 242.59 $&$ 292.91 $&$ 345.89 $&$ 398.77 $&$ {} $\\ 
\hline 
$16 			$&$ 1522.54 $&$ 1320.98 $&$ 1405.90 $&$ 1722.54 $&$ 2061.72 $&$ 2503.12 $&$ 2887.52 $&$ 3223.13 $\\ 
\hline 
\end{tabular}
\caption{Timings over the number field $L$}
\label{tbl:tps-comp-L}
\end{table}

\begin{table}
\centering
\begin{tabular}{|c|c|c|c|c|c|c|c|c|}
\hline 
$n \backslash k	$&$ 2 $&$ 4 $&$ 6 $&$ 8 $&$ 10 $&$ 12 $&$ 14 $&$ 16 $\\ 
\hline 
$4 				$&$ 0.10 $&$ 0.09 $&$ {} $&$ {} $&$ {} $&$ {} $&$ {} $&$ {} $\\ 
\hline 
$6 				$&$ 0.33 $&$ 0.31 $&$ 0.25 $&$ {} $&$ {} $&$ {} $&$ {} $&$ {} $\\ 
\hline 
$8 				$&$ 0.83 $&$ 0.83 $&$ 0.76 $&$ 0.62 $&$ {} $&$ {} $&$ {} $&$ {} $\\ 
\hline 
$10 			$&$ 1.33 $&$ 1.37 $&$ 1.32 $&$ 1.19 $&$ 1.00 $&$ {} $&$ {} $&$ {} $\\ 
\hline 
$12 			$&$ 2.33 $&$ 2.46 $&$ 2.46 $&$ 2.31 $&$ 2.06 $&$ 1.78 $&$ {} $&$ {} $\\ 
\hline 
$14 			$&$ 5.13 $&$ 5.41 $&$ 5.42 $&$ 5.20 $&$ 4.78 $&$ 4.25 $&$ 3.72 $&$ {} $\\ 
\hline 
$16 			$&$ 6.69 $&$ 7.20 $&$ 7.31 $&$ 7.17 $&$ 6.74 $&$ 6.20 $&$ 5.53 $&$ 4.96 $\\ 
\hline 
\end{tabular}
\caption{Timings using the residue field $\OL/\mathfrak{P}$}\label{tbl:tps-comp-OL}
\end{table}


\section{Conclusion}
	Given any cyclic Galois extension $K \hookrightarrow L$ provided with
an automorphism $\theta$ generating the Galois group $\Auto_K(L)$, we
can design generalized Gabidulin codes.  Cyclotomic, Kummer or
Artin-Schreier extensions are examples of extensions that fulfill the
condition which enable to design the codes.

We also have provided various useful definitions of the rank
metric. These generalized Gabidulin codes have the same properties as
their analogues in finite fields, namely they are also MRD, and can be
decoded with an adaptation of the Welch-Berlekamp algorithm with
quadratic complexity in terms of operations in $L$.  Such a code, with
parameters $[n,k,d]_r$, enables to correct up to $s_c$ column
erasures, $s_r$ row erasures and an error of rank $t$ if $s_c+s_r+2t
\leq n-k$ in both line erasure and network coding erasure models.

Of course, over an infinite field, one obstacle is the growth of
intermediate coefficients.  We can circumvent this problem computing
modulo an inert prime ideal of $\Z$, chosen large enough when the size
of the message or the error is known, by observing that a Generalized
Gabidulin code modulo a prime ideal is a classical Gabidulin code over
a finite field.

\bibliographystyle{alpha}
\bibliography{biblio}

\end{document}